\documentclass[11pt,english]{article}

\usepackage[T1]{fontenc}
\usepackage[latin9]{inputenc}
\usepackage{geometry}
\usepackage{float}
\usepackage{url}
\usepackage{enumitem}
\usepackage{amsmath}
\usepackage{amsthm}
\usepackage{amssymb}
\usepackage{graphicx}

\makeatletter

\providecommand{\tabularnewline}{\\}

 \theoremstyle{definition}
 \newtheorem*{defn*}{\protect\definitionname}
\theoremstyle{plain}
\newtheorem{thm}{\protect\theoremname}[section]
  \theoremstyle{plain}
  \newtheorem{conjecture}[thm]{\protect\conjecturename}
  \theoremstyle{plain}
  \newtheorem*{fact*}{\protect\factname}
  \theoremstyle{plain}
  \newtheorem*{cor*}{\protect\corollaryname}
  \theoremstyle{remark}
  \newtheorem*{rem*}{\protect\remarkname}
  \theoremstyle{definition}
  \newtheorem*{example*}{\protect\examplename}
  \theoremstyle{plain}
  \newtheorem{prop}[thm]{\protect\propositionname}
  \theoremstyle{remark}
  \newtheorem{claim}[thm]{\protect\claimname}
  \theoremstyle{plain}
  \newtheorem{fact}[thm]{\protect\factname}
  \theoremstyle{plain}
  \newtheorem{cor}[thm]{\protect\corollaryname}
 \newlist{casenv}{enumerate}{4}
 \setlist[casenv]{leftmargin=*,align=left,widest={iiii}}
 \setlist[casenv,1]{label={{\itshape\ \casename} \arabic*.},ref=\arabic*}
 \setlist[casenv,2]{label={{\itshape\ \casename} \roman*.},ref=\roman*}
 \setlist[casenv,3]{label={{\itshape\ \casename\ \alph*.}},ref=\alph*}
 \setlist[casenv,4]{label={{\itshape\ \casename} \arabic*.},ref=\arabic*}
  \theoremstyle{plain}
  \newtheorem{lem}[thm]{\protect\lemmaname}

\makeatother

\usepackage{babel}
  \providecommand{\claimname}{Claim}
  \providecommand{\conjecturename}{Conjecture}
  \providecommand{\corollaryname}{Corollary}
  \providecommand{\definitionname}{Definition}
  \providecommand{\examplename}{Example}
  \providecommand{\factname}{Fact}
  \providecommand{\lemmaname}{Lemma}
  \providecommand{\propositionname}{Proposition}
  \providecommand{\remarkname}{Remark}
 \providecommand{\casename}{Case}
\providecommand{\theoremname}{Theorem}

\begin{document}

\title{Improved Lower Bounds for the Fourier Entropy/Influence Conjecture
via Lexicographic Functions\global\long\def\E{{\rm \mathbf{E}}}
\global\long\def\V{{\rm \mathbf{Var}}}
\global\long\def\Pr{{\rm Pr}}
\global\long\def\H{{\bf H}}
\global\long\def\I{{\bf I}}
\global\long\def\S{{\bf S}}
\global\long\def\true{{\bf true}}
\global\long\def\false{{\bf false}}
\global\long\def\AND{{\rm And}}
\global\long\def\OR{{\rm Or}}
}

\author{Rani Hod\thanks{Department of Mathematics, Bar Ilan University, Ramat Gan, Israel.
Email: \protect\url{rani.hod@math.biu.ac.il}. The research was partially
supported by the Israel Science Foundation (grants no.~402/13 and~1612/17),
by the Binational US-Israel Science Foundation (grant no.~2014290),
and by the Coleman\textendash Soref postdoctoral fellowship.}}
\maketitle
\begin{abstract}
Every Boolean function can be uniquely represented as a multilinear
polynomial. The entropy and the total influence are two ways to measure
the concentration of its Fourier coefficients, namely the monomial
coefficients in this representation: the entropy roughly measures
their spread, while the total influence measures their average level.
The Fourier Entropy/Influence conjecture of Friedgut and Kalai from
1996 states that the entropy to influence ratio is bounded by a universal
constant $C$.

Using lexicographic Boolean functions, we present three explicit asymptotic
constructions that improve upon the previously best known lower bound
$C>6.278944$ by O'Donnell and Tan, obtained via recursive composition.
The first uses their construction with the lexicographic function
$\ell\left\langle 2/3\right\rangle $ of measure $2/3$ to demonstrate
that $C\ge4+3\log_{4}3>6.377444$. The second generalizes their construction
to biased functions and obtains $C>6.413846$ using $\ell\left\langle \Phi\right\rangle $,
where $\Phi$ is the inverse golden ratio. The third, independent,
construction gives $C>6.454784$, even for monotone functions.

Beyond modest improvements to the value of $C$, our constructions
shed some new light on the properties sought in potential counterexamples
to the conjecture.

Additionally, we prove a Lipschitz-type condition on the total influence
and spectral entropy, which may be of independent interest.
\end{abstract}

\section{\label{sec:intro}Introduction}

Let $\true=-1$ and $\false=+1$. Throughout this paper, we write
$\left[n\right]=\left\{ 1,2,\ldots,n\right\} $ and $N=2^{n}$ for
an integer $n\ge1$. It is well known that any function $f:\left\{ \true,\false\right\} ^{n}\to\mathbb{R}$
can be expressed as 
\[
f=\sum_{S\subseteq\left[n\right]}\hat{f}\left(S\right)\chi_{S},
\]
where $\chi_{S}\left(x\right)=\prod_{i\in S}x_{i}$ for $S\subseteq\left[n\right]$
are the \emph{Fourier basis} \emph{functions} and 
\[
\hat{f}\left(S\right)=\left\langle f,\chi_{S}\right\rangle =\E\left[f\left(x\right)\chi_{S}\left(x\right)\right]
\]
for $S\subseteq\left[n\right]$ are called the \emph{Fourier coefficients}
of~$f$. When $f$ is a Boolean function, i.e., $f:\left\{ \true,\false\right\} ^{n}\to\left\{ \true,\false\right\} $,
we have $\sum_{S\subseteq\left[n\right]}\hat{f}\left(S\right)^{2}=1$
by Parseval, so we can treat the Fourier coefficients' squares as
a probability distribution~$p_{f}$ on the~$N$ subsets of~$\left[n\right]$,
which we call the \emph{spectral distribution} of $f$.

The following two parameters of the function $f$ can be defined in
terms of its spectral distribution.
\begin{defn*}
The \emph{total influence} (also called average sensitivity) of a
Boolean function~$f$ is 
\[
\I\left[f\right]=\E_{S\sim p_{f}}\left[\left|S\right|\right].
\]
\end{defn*}

\begin{defn*}
The \emph{spectral entropy} of a Boolean function~$f$ is the (Shannon)
entropy of its spectral distribution
\[
\H\left[f\right]=\E_{S\sim p_{f}}\left[-\log_{2}\left(p_{f}\left(S\right)\right)\right]=-\sum_{S}p_{f}\left(S\right)\log_{2}\left(p_{f}\left(S\right)\right).
\]
\end{defn*}
In 1996 Friedgut and Kalai raised the following conjecture, known
as the Fourier Entropy/Influence (FEI) conjecture:
\begin{conjecture}[\cite{FK96}]
\label{conj:fei}There exists a universal constant $C>0$ such that
for every Boolean function $f$ with total influence $\I\left[f\right]$
and spectral entropy $\H\left[f\right]$ we have $\H\left[f\right]\le C\cdot\I\left[f\right]$.
\end{conjecture}

Conjecture~\ref{conj:fei} was verified for various families of Boolean
functions (e.g., symmetric functions~\cite{OWZ11}, random functions~\cite{DPV11},
read-once formulas~\cite{CKLS15,OT13}, decision trees of constant
average depth~\cite{WWW14}, read-$k$ decision trees for constant
$k$~\cite{WWW14}) but is still open for the class of general Boolean
functions.

\medskip{}
The rest of this paper is organized as follows. In the remainder of
Section~\ref{sec:intro} we describe past results and some rudimentary
improvements. In Section~\ref{sec:LB-1} we introduce lexicographic
functions and provide a formal proof of the approach described in
Section~\ref{sec:sequences}. In Section~\ref{sec:LB-2} we generalize
Proposition~\ref{prop:composition} to biased functions and get an
improved lower bound. In Section~\ref{sec:LB-3} we build a limit-of-limits
function that achieves an even better bound. In Section~\ref{sec:Lipschitz}
we prove a Lipschitz-type condition used throughout the paper, namely
that a small change in a Boolean function cannot result in a substantial
change to its total influence and spectral entropy.

\subsection{\label{sec:baby}A baby example and two definitions}

Here is a example of providing a lower bound on C. For $n\ge1$ consider
the function 
\[
\AND_{n}\left(x_{1},\ldots,x_{n}\right)=x_{1}\wedge x_{2}\wedge\cdots\wedge x_{n}.
\]
It satisfies $\I\left[\AND_{n}\right]=2n/N$ and $\H\left[\AND_{n}\right]\approx8\left(n-1+1/\ln4\right)/N,$\footnote{More precisely, $0<8\left(n-1+1/\ln4\right)/N-\H\left[\AND_{n}\right]<12n/N^{2}$
for all $n\ge1$.} so any constant $C$ in Conjecture~\ref{conj:fei} must satisfy
\[
C\ge\H\left[\AND_{n}\right]/\I\left[\AND_{n}\right]\approx4-\frac{4}{n}\left(1-1/\ln4\right).
\]
This is true for every $n$, so by taking $n\to\infty$ we establish
that $C\ge4$. 
\begin{defn*}
A Boolean function $f$ is called \emph{monotone} if changing an input
bit $x_{i}$ from $\false$ to $\true$ cannot change the output $f\left(x\right)$
from $\true$ to $\false$.
\end{defn*}
\begin{fact*}
A Boolean function is monotone if and only if it can be expressed
as a formula combining variables using conjunctions ($\wedge$) and
disjunctions ($\vee$) only, with no negations.
\end{fact*}
\begin{defn*}
Let $f$ be a Boolean function on $n$ variables. The \emph{dual }function
of~$f$, denoted~$f^{\dagger}$, is defined as
\[
f^{\dagger}\left(x_{1},\ldots,x_{n}\right)=\neg f\left(\neg x_{1},\ldots\neg x_{n}\right).
\]
\end{defn*}
\begin{fact*}
For all $S\subseteq\left[n\right]$ we have $\hat{f^{\dagger}}\left(S\right)=\left(-1\right)^{\left|S\right|+1}\hat{f}\left(S\right)$. 
\end{fact*}
\begin{cor*}
The spectral distributions $p_{f}$ and $p_{f^{\dagger}}$ are identical;
in particular, $\I\left[f^{\dagger}\right]=\I\left[f\right]$ and
$\H\left[f^{\dagger}\right]=\H\left[f\right]$. But $\E\left[f^{\dagger}\right]=-\E\left[f\right]$
and $\Pr\left[f^{\dagger}\left(x\right)=\true\right]=1-\Pr\left[f\left(x\right)=\true\right]$.
\end{cor*}
\begin{rem*}
If~$f$ is monotone then~$f^{\dagger}$ is monotone too. Furthermore,
given a monotone formula computing~$f$, the formula obtained by
swapping conjunctions and disjunctions computes~$f^{\dagger}$.
\end{rem*}
\begin{example*}
The dual of $\OR_{n}\left(x_{1},\ldots,x_{n}\right)=x_{1}\vee x_{2}\vee\cdots\vee x_{n}$
is $\AND_{n}$.
\end{example*}

\subsection{\label{sec:past}Past results and preliminary improvements}

The current best lower bound on $C$ was achieved by O'Donnell and
Tan~\cite{OT13}. Using recursive composition they showed the following
bound:
\begin{prop}
\label{prop:composition}Let $g$ be a balanced Boolean function such
that $\H\left[g\right]>0$. Then any constant~$C$ in Conjecture~\ref{conj:fei}
satisfies $C\ge\H\left[g\right]/\left(\I\left[g\right]-1\right).$
\end{prop}

\begin{rem*}
Any balanced Boolean function $g$ has $\I\left[g\right]\ge1$ since
$p_{g}\left(\varnothing\right)=\E\left[g\right]=0$; in case of equality
we must have $g=\chi_{\left\{ i\right\} }$ for some $i\in\left[n\right]$
and thus $p_{g}$ is supported on a single set $S=\left\{ i\right\} $
and its spectral entropy is zero.
\end{rem*}
By presenting a function on $6$ variables with total influence $I=13/8=1.625$
and entropy $H>3.92434$, they established that $C>3.92434/\frac{5}{8}=6.278944$.
Although the specific function presented in~\cite{OT13} happens
to be biased, their result stands as there exists a balanced Boolean
function $g_{3}$ on 6 variables with the same total influence and
entropy:
\[
g_{3}\left(x_{1},x_{2},x_{3},x_{4},x_{5},x_{6}\right)=\left(x_{1}\vee x_{2}\right)\wedge\left(x_{3}\vee x_{4}\right)\wedge\left(x_{1}\vee x_{3}\vee x_{5}\right)\wedge\left(x_{3}\vee x_{5}\vee x_{6}\right).
\]
A slight improvement can be achieved by modifying the last clause
of $g_{3}$. Indeed,
\[
g_{3}'\left(x_{1},x_{2},x_{3},x_{4},x_{5},x_{6}\right)=\left(x_{1}\vee x_{2}\right)\wedge\left(x_{3}\vee x_{4}\right)\wedge\left(x_{1}\vee x_{3}\vee x_{5}\right)\wedge\left(x_{2}\vee x_{4}\vee x_{6}\right)
\]
is balanced too, with the same total influence $\I\left[g_{3}'\right]=\I\left[g_{3}\right]=13/8$
and a slightly higher entropy $\H\left[g_{3}'\right]>3.9669$, so
we have $C>3.9669/\frac{5}{8}=6.34704$.

Moving to balanced functions on 8 variables, we find a monotone function
$g_{4}$ that provides a better lower bound:
\[
g_{4}\left(x_{1},\ldots,x_{8}\right)=\left(x_{1}\vee x_{2}\right)\wedge\left(x_{3}\vee x_{4}\right)\wedge\left(x_{3}\vee x_{5}\vee x_{6}\right)\wedge\left(x_{1}\vee x_{3}\vee x_{5}\vee x_{7}\right)\wedge\left(x_{3}\vee x_{5}\vee x_{7}\vee x_{8}\right)
\]
with $\I\left[g_{4}\right]=53/32$ and $\H\left[g_{4}\right]>4.16885$
yields $C>6.35253$.

A further search discovers a slightly superior function:
\[
g_{4}'\left(x_{1},\ldots,x_{8}\right)=\left(x_{1}\vee x_{2}\right)\wedge\left(x_{3}\vee x_{4}\right)\wedge\left(x_{3}\vee x_{5}\vee x_{6}\right)\wedge\left(x_{1}\vee x_{3}\vee x_{5}\vee x_{7}\right)\wedge\left(x_{2}\vee x_{3}\vee x_{6}\vee x_{8}\right)
\]
with $\I\left[g_{4}'\right]=\I\left[g_{4}\right]=53/32$ and $\H\left[g_{4}'\right]>4.17635$
achieves $C>6.36396$.

\subsection{\label{sec:sequences}Sequences of balanced monotone functions}

Staring at $g_{3}$, $g_{4}$ and $g_{4}'$ for a moment (but not
$g_{3}'$), we may see a common property: $x_{3}$ appears in all
clauses except the first. Let us rewrite $g_{3}$ and $g_{4}$ in
a slightly different form:
\begin{eqnarray*}
g_{3}\left(x_{1},\ldots,x_{6}\right) & = & \left(x_{1}\vee x_{2}\right)\wedge\left(x_{3}\vee\left(x_{4}\wedge\left(x_{5}\vee\left(x_{6}\wedge x_{1}\right)\right)\right)\right),\\
g_{4}\left(x_{1},\ldots,x_{8}\right) & = & \left(x_{1}\vee x_{2}\right)\wedge\left(x_{3}\vee\left(x_{4}\wedge\left(x_{5}\vee\left(x_{6}\wedge\left(x_{7}\vee\left(x_{8}\wedge x_{1}\right)\right)\right)\right)\right)\right).
\end{eqnarray*}
This generalizes easily to a sequence $\left(g_{m}\right)_{m\ge1}$
of balanced (to be shown below) monotone Boolean functions: 
\[
g_{m}\left(x_{1},\ldots,x_{2m}\right)=\left(x_{1}\vee x_{2}\right)\wedge\left(x_{3}\vee\left(x_{4}\wedge\left(x_{5}\vee\cdots\left(x_{2m-1}\vee\left(x_{2m}\wedge x_{1}\right)\right)\cdots\right)\right)\right)
\]
whose first two members are
\begin{eqnarray*}
g_{1}\left(x_{1},x_{2}\right) & = & \left(x_{1}\vee x_{2}\right)\wedge x_{1}=x_{1},\\
g_{2}\left(x_{1},x_{2},x_{3},x_{4}\right) & = & \left(x_{1}\vee x_{2}\right)\wedge\left(x_{3}\vee\left(x_{4}\wedge x_{1}\right)\right)=\left(x_{1}\vee x_{2}\right)\wedge\left(x_{3}\vee x_{4}\right)\wedge\left(x_{1}\vee x_{3}\right).
\end{eqnarray*}
Denote by $C_{m}=\H\left[g_{m}\right]/\left(\I\left[g_{m}\right]-1\right)$
the lower bound on $C$ implied by $g_{m}$. The first fifteen members
of the sequence are explored in Table~\ref{tbl:params}. Note how
even $C_{2}=6$ is much better than the $C\ge4$ bound of Subsection~\ref{sec:baby}.
\begin{table}[H]
\begin{centering}
\begin{tabular}{|c|c|l|l|l|}
\hline 
$m$ & $n$ & $\I\left[g_{m}\right]$ & $\H\left[g_{m}\right]$ & $C_{m}$\tabularnewline
\hline 
\hline 
1 & 2 & 1 & 0 & (not defined)\tabularnewline
\hline 
2 & 4 & $3/2=1.5$ & 3 & 6\tabularnewline
\hline 
3 & 6 & $13/8=1.625$ & $>3.92434$ & $>6.27894$\tabularnewline
\hline 
4 & 8 & $53/32=1.65625$ & $>4.16885$ & $>6.35253$\tabularnewline
\hline 
5 & 10 & $213/128=1.6640625$ & $>4.23087$ & $>6.37119$\tabularnewline
\hline 
6 & 12 & $853/512=\left(5-2^{-9}\right)/3$ & $>4.24643$ & $>6.37588$\tabularnewline
\hline 
7 & 14 & $3413/2048=\left(5-2^{-11}\right)/3$ & $>4.25033$ & $>6.37705$\tabularnewline
\hline 
8 & 16 & $13653/2^{13}=\left(5-2^{-13}\right)/3$ & $>4.25130$ & $>6.37734$\tabularnewline
\hline 
9 & 18 & $54613/2^{15}=\left(5-2^{-15}\right)/3$ & $>4.25154$ & $>6.37741$\tabularnewline
\hline 
10 & 20 & $\left(5-2^{-17}\right)/3$ & $>4.251608$ & $>6.377437$\tabularnewline
\hline 
11 & 22 & $\left(5-2^{-19}\right)/3$ & $>4.251624$ & $>6.3774422$\tabularnewline
\hline 
12 & 24 & $\left(5-2^{-21}\right)/3$ & $>4.2516278$ & $>6.3774433$\tabularnewline
\hline 
13 & 26 & $\left(5-2^{-23}\right)/3$ & $>4.25162885$ & $>6.37744365$\tabularnewline
\hline 
14 & 28 & $\left(5-2^{-25}\right)/3$ & $>4.25162908$ & $>6.37744372$\tabularnewline
\hline 
15 & 30 & $\left(5-2^{-27}\right)/3$ & $>4.251629147$ & $>6.377443745$\tabularnewline
\hline 
\end{tabular}
\par\end{centering}
\caption{\label{tbl:params}Parameters of the sequence $g_{m}$ for $m\le15$}
\end{table}
The three sequences seem to be increasing and bounded, so let us denote
their respective hypothetical limits by $I_{*}$, $H_{*}$ and $C_{*}$.
If indeed $\I\left[g_{m}\right]=\frac{1}{3}\left(5-2^{3-2m}\right)$
for all $m\in\mathbb{N}$ then $I_{*}=5/3$. A prescient guess for
the value of $H_{*}$ could be 
\[
H_{*}=\frac{8}{3}+\log_{2}3>4.251629167,
\]
for which we would get 
\[
C_{*}=H_{*}/\left(2/3\right)=4+3\log_{4}3>6.377443751
\]
as a lower bound for $C$. We will verify this guess in Section~\ref{sec:LB-1}.\medskip{}

Recall that $g_{3}'$ and $g_{4}'$ gave rise to better lower bounds,
respectively, than $g_{3}$ and $g_{4}$. It is tempting perhaps to
consider a generalization $\left(g_{m}'\right)_{m\ge3}$, define $C_{m}'$
accordingly and examine the hypothetical limits $I_{*}'$, $H_{*}'$
and $C_{*}'$. It is indeed possible to do so, and we get $\I\left[g_{m}'\right]=\I\left[g_{m}'\right]$
while $\H\left[g_{m}'\right]>\H\left[g_{m}\right]$, making $C'_{m}>C_{m}$.
Nevertheless, $\H\left[g_{m}'\right]$ and $C'_{m}$ seem to converge
towards the same $H_{*}$ and $C_{*}$, respectively, so there is
no real benefit in pursuing this further.

\medskip{}
It remains to verify that $g_{m}$ is indeed balanced for all $m\ge1$.
Let us write it as 
\[
g_{m}\left(x_{1},\ldots,x_{2m}\right)=\left(x_{1}\vee x_{2}\right)\wedge G_{m}\left(x_{3},x_{4},\ldots,x_{2m},x_{1}\right),
\]
where $G_{m}\left(y_{1},y_{2},\ldots,y_{2m-1}\right)$ is defined
recursively via 
\begin{eqnarray*}
G_{1}\left(y_{1}\right) & = & y_{1},\\
G_{m+1}\left(y_{1},\ldots,y_{2m+1}\right) & = & y_{1}\vee\left(y_{2}\wedge G_{m}\left(y_{3},\ldots,y_{2m+1}\right)\right).
\end{eqnarray*}
\begin{rem*}
The function $G_{m}$ belongs to a class of monotone Boolean functions
called lexicographic functions, as we will see in Section~\ref{sec:lexicographic}.
\end{rem*}
For simplicity of notation, we abbreviate and write $\Pr\left[f\left(x\right)\right]$
or even $\Pr\left[f\right]$ to denote $\Pr\left[f\left(x\right)=\true\right]$.
Since
\[
\Pr\left[g_{m}\left(x\right)\right]=\Pr\left[x_{1}\vee x_{2}\right]\cdot\Pr\left[G_{m}\left(x_{3},\ldots,x_{2m},x_{1}\right)\mid x_{1}\vee x_{2}\right],
\]
to prove $\Pr\left[g_{m}\right]=\frac{1}{2}$, it suffices to verify
the following (see Appendix~\ref{apx:boring} for the calculation):
\begin{claim}
\label{clm:Gm-2-3-conditioned}For all $m\ge1$ we have $\Pr\left[G_{m}\left(x_{3},\ldots,x_{2m},x_{1}\right)\mid x_{1}\vee x_{2}\right]=2/3$.
\end{claim}

\section{\label{sec:LB-1}A Tale of Two Thirds}

Although each of $\left(C_{m}\right)_{m=2}^{15}$ from Table~\ref{tbl:params}
is a valid, explicit lower bound on $C$, the asymptotic discussion
in Subsection~\ref{sec:sequences} was more of a wishful thinking
rather than a mathematically sound statement.

In this section we explore the class of lexicographic functions, develop
tools to compute total influence and spectral entropy, and then rigorously
calculate $I_{*}$, $H_{*}$ and $C_{*}$.

\subsection{\label{sec:lexicographic}Lexicographic functions}
\begin{defn*}
Fix integers $n\ge1$ and $0\le s\le N$. Denote by $L_{n}\left\langle s\right\rangle \subseteq\left\{ \true,\false\right\} ^{n}$
the initial segment of cardinality $s$ (with respect to the lexicographic
order on $\left\{ \true,\false\right\} ^{n}$), and denote by $\ell_{n}\left\langle s\right\rangle :\left\{ \true,\false\right\} ^{n}\to\left\{ \true,\false\right\} $
its characteristic function 
\[
\ell_{n}\left\langle s\right\rangle \left(x\right)=\begin{cases}
\true, & x\in L_{n}\left\langle s\right\rangle ;\\
\false, & x\not\in L_{n}\left\langle s\right\rangle .
\end{cases}
\]
\end{defn*}
\begin{fact*}
We have $\Pr\left[\ell_{n}\left\langle s\right\rangle \right]=\left|L_{n}\left\langle s\right\rangle \right|/N=s/N$
and $\E\left[\ell_{n}\left\langle s\right\rangle \right]=1-2s/N$.
\end{fact*}

\begin{fact*}
The function $\ell_{n}\left\langle s\right\rangle $ is monotone and
its dual is $\left(\ell_{n}\left\langle s\right\rangle \right)^{\dagger}=\ell_{n}\left\langle N-s\right\rangle $.
\end{fact*}
\begin{example*}
$\ell_{n}\left\langle 1\right\rangle =\AND_{n}$ and $\ell_{n}\left\langle N-1\right\rangle =\OR_{n}$.
\end{example*}
\begin{fact*}
If $s$ is even then $\ell_{n}\left\langle s\right\rangle $ is isomorphic
to $\ell_{n-1}\left\langle s/2\right\rangle $ (when the latter is
extended from $n-1$ to $n$ variables by adding an influenceless
variable).
\end{fact*}
Let $0<s<N$ be an odd integer, and let $s_{1}s_{2}\cdots s_{n}$
be its binary representation, where~$s_{1}$ is the most significant
bit and $s_{n}=1$ is the least significant bit. Denote the corresponding
$\left\{ \true,\false\right\} ^{n}$ representation of~$s$ by~$\vec{s}=\left(\left(-1\right)^{s_{1}},\ldots,\left(-1\right)^{s_{n}}\right)$.

By definition, to determine the value of~$\ell_{n}\left\langle s\right\rangle $
for an input~$x$, we need to compare~$x$ with~$\vec{s}$ element
by element. This gives a neat formula for~$\ell_{n}\left\langle s\right\rangle $:
\begin{equation}
\ell_{n}\left\langle s\right\rangle \left(x_{1},\ldots,x_{n}\right)=x_{1}\diamond_{1}\left(x_{2}\diamond_{2}\left(x_{3}\diamond_{3}\cdots\left(x_{n-1}\diamond_{n-1}x_{n}\right)\cdots\right)\right),\label{eq:lexicographic-formula}
\end{equation}
where $\diamond_{i}=\begin{cases}
\wedge, & s_{i}=0;\\
\vee, & s_{i}=1.
\end{cases}$
\begin{rem*}
The formula~(\ref{eq:lexicographic-formula}) shows that every monotone
decision list, i.e., a monotone decision tree consisting of a single
path, is isomorphic to a lexicographic function.
\end{rem*}
From~(\ref{eq:lexicographic-formula}) we derive an important property
of lexicographic functions.
\begin{fact}
\label{fact:exp-decreasing-influences}For $k\in\left[n\right]$,
the value of $\ell_{n}\left\langle s\right\rangle $$\left(x\right)$
only depends on $x_{k}$ with probability $2^{1-k}$; that is, when
$x_{i}=\left(-1\right)^{s_{i}}$ for all $i<k$.
\end{fact}

\begin{rem*}
This can be interpreted as saying that the average decision tree complexity
of $\ell_{n}\left\langle s\right\rangle $ is $2-\left(n+2\right)/N$.

\medskip{}

We extend the definition of lexicographic functions by writing $\ell_{n}\left\langle \mu\right\rangle =\ell_{n}\left\langle \left\lfloor \mu N\right\rfloor \right\rangle $
for some $0\le\mu\le1$. Note that $\left\lfloor \mu N\right\rfloor $
is not necessarily odd, so the effective number of variables can be
smaller.
\end{rem*}
\begin{example*}
For any $n\ge2$ we have $\ell_{n}\left\langle 3/4\right\rangle \left(x\right)=x_{1}\vee x_{2}$;
that is, $\ell_{n}\left\langle 3/4\right\rangle =\OR_{2}$.
\end{example*}
\begin{example*}
For $n=2m-1$, we have $\ell_{n}\left\langle 2/3\right\rangle =\ell_{n}\left\langle \left\lfloor 2N/3\right\rfloor \right\rangle =\ell_{n}\left\langle \left(2N-1\right)/3\right\rangle $.
Observe that the binary representation of the odd integer $s=\left(2N-1\right)/3$
has $s_{i}=i\bmod2$ for $i\in\left[n\right]$ and thus
\[
\ell_{n}\left\langle 2/3\right\rangle \left(x\right)=x_{1}\vee\left(x_{2}\wedge\left(x_{3}\vee\left(x_{4}\wedge\cdots\left(x_{2m-2}\wedge x_{2m-1}\right)\cdots\right)\right)\right),
\]
that is, $\ell_{n}\left\langle 2/3\right\rangle =G_{m}$.
\end{example*}
Fix $0<\mu<1$ and consider the sequence $\left(\ell_{n}\left\langle \mu\right\rangle \right)_{n\ge1}$.
Whenever $\mu$ is a dyadic rational\footnote{That is, a rational number of the form $a/2^{b}$.},
$\ell_{n}\left\langle \mu\right\rangle $ converges to a fixed function
$\ell\left\langle \mu\right\rangle $ (e.g., $\ell\left\langle 3/4\right\rangle =\OR_{2}$
in the example above). We would like to consider the limit object
$\ell\left\langle \mu\right\rangle =\lim_{n}\ell_{n}\left\langle \mu\right\rangle $
for other values of $\mu$ as well.

It may sound intimidating; after all, $\ell\left\langle \mu\right\rangle :\left\{ \true,\false\right\} ^{\mathbb{N}}\to\left\{ \true,\false\right\} $
is a Boolean function on $\aleph_{0}$ variables, which is quite a
lot. Nevertheless, by Fact~\ref{fact:exp-decreasing-influences},
$\ell\left\langle \mu\right\rangle $ only reads \emph{two} input
bits on average.

Moreover, we care about the total influence and spectral entropy of
functions. By Lemmata~\ref{lem:influence-distance} and~\ref{lem:entropy-distance}
from Section~\ref{sec:Lipschitz}, $\I\left[\ell_{n}\left\langle \mu\right\rangle \right]\xrightarrow{n\to\infty}\I\left[\ell\left\langle \mu\right\rangle \right]$
and $\H\left[\ell_{n}\left\langle \mu\right\rangle \right]\xrightarrow{n\to\infty}\H\left[\ell\left\langle \mu\right\rangle \right]$.
Indeed, $\ell_{n}\left\langle \mu\right\rangle $ differs from $\ell_{n-1}\left\langle \mu\right\rangle $
(when considering the latter as a function on $n$ variables by adding
an influenceless variable) in at most one place, and thus $\left(\I\left[\ell_{n}\left\langle \mu\right\rangle \right]\right)_{n\ge1}$
and $\left(\H\left[\ell_{n}\left\langle \mu\right\rangle \right]\right)_{n\ge1}$
are Cauchy sequences.

Needless to say, $\Pr\left[\ell_{n}\left\langle \mu\right\rangle \right]=\left\lfloor \mu N\right\rfloor /N\xrightarrow{n\to\infty}\mu=\Pr\left[\ell\left\langle \mu\right\rangle \right]$.

\medskip{}

An even stronger statement holds (but will not be used or proved here):
the spectral distributions of $\ell_{n}\left\langle \mu\right\rangle $
converge in distribution to a limit distribution $p_{\mu}$, which
we call the spectral distribution of $\ell\left\langle \mu\right\rangle $.
Note that $p_{\mu}$ is supported on \emph{finite} subsets of $\mathbb{N}$.
The expected cardinality and the entropy of $S\sim p_{\mu}$ are $\I\left[\ell\left\langle \mu\right\rangle \right]$
and $\H\left[\ell\left\langle \mu\right\rangle \right]$ respectively.

\subsection{Total influence and lexicographic functions}

The edge isoperimetric inequality in the discrete cube (by Harper~\cite{Harp64},
with an addendum by Bernstein~\cite{Bern67}, and independently Lindsey~\cite{Lind64})
gives a lower bound on the total influence of Boolean functions.
\begin{thm}
\label{thm:edge-isoperimetric-ineq}Let $f$ be a Boolean function
with $\Pr\left[f\right]=\mu\le\frac{1}{2}$. Then $\I\left[f\right]\ge-2\mu\log_{2}\mu$.
\end{thm}

In fact, they proved that lexicographic functions are the minimizers
of total influence. 
\begin{thm}
\label{thm:edge-isoperimetric-ineq-harper}Fix integers $n\ge1$ and
$s\le N/2$ and let $f$ be a Boolean function on $n$ variables with
$\Pr\left[f\right]=s/N$. Then $\I\left[f\right]\ge\I\left[\ell_{n}\left\langle s\right\rangle \right]$.
\end{thm}

\begin{rem*}
Theorem~\ref{thm:edge-isoperimetric-ineq-harper} explains our interest
in lexicographic functions: when seeking a function $f$ with large
entropy/influence ratio $\H\left[f\right]/\I\left[f\right]$, it makes
sense to minimize $\I\left[f\right]$.
\end{rem*}
In~\cite{Hart76}, Hart exactly computed the total influence of lexicographic
functions:
\begin{prop}[{\cite[Theorem 1.5]{Hart76}}]
\label{prop:lex-influence} Fix integers $n\ge1$ and $0\le s\le N$.
Then 
\[
\I\left[\ell_{n}\left\langle s\right\rangle \right]=\frac{2sn}{N}-\frac{4}{N}\sum_{x=0}^{s-1}wt\left(x\right),
\]
 where $wt\left(x\right)$ is the Hamming weight of $x$.
\end{prop}

Let us rephrase Proposition~\ref{prop:lex-influence} a bit.
\begin{claim}
\label{clm:lex-influence}Let $s=\sum_{i=0}^{t}N/2^{k_{i}}$, where
$1\le k_{0}<k_{1}<\cdots<k_{t}$ are the locations of $1$ in the
binary representation of $s$. Then $\I\left[\ell_{n}\left\langle s\right\rangle \right]=\sum_{i=0}^{t}\left(k_{i}-2i\right)2^{1-k_{i}}$.
\end{claim}

\begin{proof}
By induction on $t$. For details see Appendix~\ref{apx:boring}.
\end{proof}
\begin{example*}
For $s=N/2^{k}$ we get $\I\left[\AND_{k}\right]=\I\left[\ell_{n}\left\langle N/2^{k}\right\rangle \right]=k2^{1-k}$,
demonstrating the tightness of Theorem~\ref{thm:edge-isoperimetric-ineq}.
\end{example*}
\begin{cor}
\label{cor:lex-influence}Let $\mu=\sum_{i=0}^{\infty}2^{-k_{i}}$,
where $1\le k_{0}<k_{1}<\cdots$ are the locations of $1$ in the
binary representation of $\mu$.\footnote{To be read as a finite sum when $\mu$ is a dyadic rational.}
Then $\I\left[\ell\left\langle \mu\right\rangle \right]=\sum_{i=0}^{\infty}\left(k_{i}-2i\right)2^{1-k_{i}}$.
\end{cor}

This leads to the following observation:
\begin{fact*}
For any $0\le\mu\le1$ we have 
\begin{equation}
\I\left[\ell\left\langle \frac{1}{2}\pm\frac{\mu}{4}\right\rangle \right]=2\cdot2^{-1}+\sum_{i=1}^{\infty}\left(\left(k_{i-1}+2\right)-2i\right)2^{1-\left(k_{i-1}+2\right)}=1+\frac{1}{4}\I\left[\ell\left\langle \mu\right\rangle \right].\label{eq:lex-influence}
\end{equation}
\end{fact*}
\begin{example*}
For $\mu=\frac{2}{3}$ we have 
\[
\I\left[\ell\left\langle 2/3\right\rangle \right]=\I\left[\ell\left\langle 1/2+1/6\right\rangle \right]=1+\frac{1}{4}\I\left[\ell\left\langle 2/3\right\rangle \right],
\]
hence $\I\left[\ell\left\langle 2/3\right\rangle \right]=\frac{4}{3}$.
By duality we have $\I\left[\ell\left\langle 1/3\right\rangle \right]=\frac{4}{3}$
as well.
\end{example*}
\begin{rem*}
Compare the bound $\frac{2}{3}\log_{2}3\approx1.05664$ obtained for
$\mu=\frac{1}{3}$ from Theorem~\ref{thm:edge-isoperimetric-ineq}
to $\I\left[\ell\left\langle 1/3\right\rangle \right]=4/3\approx1.3333$
computed above. In fact, Theorem~\ref{thm:edge-isoperimetric-ineq}
is only tight when $\mu$ is a power of two. 
\end{rem*}
Four-thirds is actually the maximum influence attainable by any lexicographic
function, as the following claim shows:
\begin{claim}
\label{clm:influence-max}For all $0\le\mu\le1$ we have $\I\left[\ell\left\langle \mu\right\rangle \right]\le\frac{4}{3}$.
\end{claim}

\begin{proof}
\begin{casenv}
\item $\mu<\frac{1}{4}$. Writing $\mu=\sum_{i}2^{-k_{i}}$, we have $k_{i}\ge i+3$
for all $i\ge0$. Moreover, we cannot have $k_{i}=i+3$ for all $i$
since $\mu<\sum_{j=3}^{\infty}2^{-j}=\frac{1}{4}$. Denote by $j$
the minimal $i$ for which $k_{i}>i+3$. Now, by Corollary~\ref{cor:lex-influence},
\[
\I\left[\ell\left\langle \mu\right\rangle \right]=\sum_{i=0}^{\infty}\left(k_{i}-2i\right)2^{1-k_{i}}=\cdots\le1+2^{-j-2}\left(1+j\right)\le\frac{5}{4}<\frac{4}{3},
\]
where the full calculation is in Appendix~\ref{apx:boring}.
\item $\mu>\frac{3}{4}$. Then $\I\left[\ell\left\langle \mu\right\rangle \right]=\I\left[\left(\ell\left\langle \mu\right\rangle \right)^{\dagger}\right]=\I\left[\ell\left\langle 1-\mu\right\rangle \right]\le\frac{5}{4}<\frac{4}{3}.$
\item $\frac{1}{4}\le\mu\le\frac{3}{4}$. Since $\I\left[\ell\left\langle \mu\right\rangle \right]$
is a continuous function of $\mu$, it has a maximum in the closed
interval $\left[\frac{1}{4},\frac{3}{4}\right]$, obtained at $\mu=\mu_{0}$.\footnote{If the maximum is attained multiple times, pick one arbitrarily.}
If $\I\left[\ell\left\langle \mu_{0}\right\rangle \right]=\frac{4}{3}+\epsilon$
for some $\epsilon>0$ then for $\mu_{1}=\left|4\mu_{0}-2\right|$
we have 
\[
\I\left[\ell\left\langle \mu_{1}\right\rangle \right]=4\left(\I\left[\ell\left\langle \mu_{0}\right\rangle \right]-1\right)=\frac{4}{3}+4\epsilon>\I\left[\ell\left\langle \mu_{0}\right\rangle \right],
\]
contradicting either the choice of $\mu_{0}$ or one of the two previous
cases.\qedhere
\end{casenv}
\end{proof}
\begin{rem*}
We have $\I\left[\ell\left\langle \mu\right\rangle \right]=\frac{4}{3}$
for other values of $\mu$ besides $\mu=\frac{1}{3}$ and $\mu=\frac{2}{3}$,
e.g., 
\[
{\textstyle \I\left[\ell\left\langle 7/12\right\rangle \right]=\I\left[\ell\left\langle \left(2+1/3\right)/4\right\rangle \right]=1+\frac{1}{4}\I\left[\ell\left\langle 1/3\right\rangle \right]=\frac{4}{3}}.
\]
\end{rem*}

\subsection{\label{sec:composition}Disjoint composition}

We now present the main tool we use to compute total influence and
spectral entropy for our construction.
\begin{defn*}
For two Boolean functions $f_{1}$ and $f_{2}$ on $n_{1}$ and $n_{2}$
variables, resp., define the Boolean functions on $n=n_{1}+n_{2}$
variables $f_{1}\sqcap f_{2}$ and $f_{1}\sqcup f_{2}$ as 
\begin{eqnarray*}
\left(f_{1}\sqcap f_{2}\right)\left(x_{1},x_{2},\ldots,x_{n}\right) & = & f_{1}\left(x_{1},\ldots,x_{n_{1}}\right)\wedge f_{2}\left(x_{n_{1}+1},\ldots,x_{n}\right);\\
\left(f_{1}\sqcup f_{2}\right)\left(x_{1},x_{2},\ldots,x_{n}\right) & = & f_{1}\left(x_{1},\ldots,x_{n_{1}}\right)\vee f_{2}\left(x_{n_{1}+1},\ldots,x_{n}\right),
\end{eqnarray*}
and denote by $\iota=\ell\left\langle 1/2\right\rangle $ the one
variable identity function. 
\end{defn*}
\begin{rem*}
The class of functions built using $\iota$, $\sqcap$, and $\sqcup$
is called read-once monotone formulas. By~(\ref{eq:lexicographic-formula})
every lexicographic function is a read-once monotone formulas. 

As mentioned in the introduction, it was shown by~\cite{CKLS15,OT13}
that read-once formulas satisfy Conjecture~\ref{conj:fei} with the
constant $C\le10$.
\end{rem*}
\begin{defn*}
Let $h:\left[0,1\right]\to\left[0,1\right]$ be the binary entropy
function, defined by 
\[
h\left(p\right)=-p\log_{2}p-\left(1-p\right)\log_{2}\left(1-p\right)
\]
for $0<p<1$ and $h\left(0\right)=h\left(1\right)=0$. We also make
extensive use of its variant 
\[
\tilde{h}\left(p\right)=h\left(4p\left(1-p\right)\right)=h(\left(1-2p\right)^{2}).
\]
\end{defn*}
\begin{fact*}
Both $h$ and $\tilde{h}$ are symmetric about $p=\frac{1}{2}$.
\end{fact*}
The following proposition is an easy corollary of~\cite[Lemmata 5.7 and 5.8]{CKLS15}.
Alternatively, it is a special case of Lemma~\ref{lem:composition}
in Section~\ref{sec:LB-2}, which is an adaptation of \cite[Proposition 3.2]{OT13}.
\begin{prop}
\label{prop:sqcap}Let $f_{1}$ and $f_{2}$ be Boolean functions
and let $p_{i}=\Pr\left[f_{i}\right]$ for $i=1,2$. Then
\begin{eqnarray*}
\I\left[f_{1}\sqcap f_{2}\right] & = & p_{2}\I\left[f_{1}\right]+p_{1}\I\left[f_{2}\right];\\
\H\left[f_{1}\sqcap f_{2}\right] & = & p_{2}\left(\H\left[f_{1}\right]-\tilde{h}\left(p_{1}\right)\right)+p_{1}\left(\H\left[f_{2}\right]-\tilde{h}\left(p_{2}\right)\right)+\psi\left(p_{1},p_{2}\right),
\end{eqnarray*}
where
\[
\psi\left(p,q\right)=\tilde{h}\left(pq\right)+4pq\left(h\left(p\right)+h\left(q\right)-h\left(pq\right)\right).
\]
\end{prop}

\begin{rem*}
Via the De Morgan equality $f_{1}\sqcup f_{2}=(f_{1}^{\dagger}\sqcap f_{2}^{\dagger})^{\dagger}$,
this also yields 
\begin{eqnarray*}
\I\left[f_{1}\sqcup f_{2}\right] & = & \left(1-p_{2}\right)\I\left[f_{1}\right]+\left(1-p_{1}\right)\I\left[f_{2}\right];\\
\H\left[f_{1}\sqcup f_{2}\right] & = & \left(1-p_{2}\right)\left(\H\left[f_{1}\right]-\tilde{h}\left(p_{1}\right)\right)+\left(1-p_{1}\right)\left(\H\left[f_{2}\right]-\tilde{h}\left(p_{2}\right)\right)+\psi\left(1-p_{1},1-p_{2}\right).
\end{eqnarray*}
\end{rem*}
Proposition~\ref{prop:sqcap} gets simplified significantly when
one of the functions is balanced, using the following observation
(see Appendix~\ref{apx:boring} for the calculation):
\begin{claim}
\label{clm:psi-balanced}Let $0<p<1$. Then $\psi\left(p,1/2\right)=2h\left(p\right)$.
\end{claim}

\begin{cor}
\label{cor:sqcap-iota}Let $f$ be a Boolean function and let $p=\Pr\left[f\right]$.
Then
\begin{eqnarray*}
\I\left[f\sqcap\iota\right] & = & \frac{1}{2}\I\left[f\right]+p,\\
\I\left[f\sqcup\iota\right] & = & \frac{1}{2}\I\left[f\right]+1-p,
\end{eqnarray*}
and 
\[
\H\left[f\sqcap\iota\right]=\H\left[f\sqcup\iota\right]=\frac{1}{2}\H\left[f\right]-\frac{1}{2}\tilde{h}\left(p\right)+2h\left(p\right).
\]
\end{cor}

\subsection{\label{sec:LB-1-computation}A first lower bound}

We could use Claim~\ref{clm:lex-influence} to compute the total
influence of $G_{m}=\ell_{2m-1}\left\langle 2/3\right\rangle $, but
we also need its spectral entropy, so we use its recursive definition
and Corollary~\ref{cor:sqcap-iota}. Since we are interested in asymptotics,
we prefer working directly with $G=\ell\left\langle 2/3\right\rangle $,
which satisfies the ``equation'' $G=\iota\sqcup\left(\iota\sqcap G\right)$. 

We already know $\I\left[G\right]=\I\left[\ell\left\langle 2/3\right\rangle \right]=\frac{4}{3}$,
whereas for the entropy we have 
\begin{align*}
\H\left[G\right] & =\H\left[\iota\sqcup\left(\iota\sqcap G\right)\right]=\frac{1}{2}\H\left[\iota\sqcap G\right]-\frac{1}{2}\tilde{h}\left(1/3\right)+2h\left(1/3\right)\\
 & =\frac{1}{2}\left(\frac{1}{2}\H\left[G\right]-\frac{1}{2}\tilde{h}\left(2/3\right)+2h\left(2/3\right)\right)-\frac{1}{2}\tilde{h}\left(1/3\right)+2h\left(1/3\right)\\
 & =\frac{1}{4}\H\left[G\right]+3h\left(1/3\right)-\frac{3}{4}\tilde{h}\left(2/3\right)
\end{align*}
and we can solve for 
\[
\H\left[G\right]=\frac{4}{3}\left(3h\left(1/3\right)-\frac{3}{4}h\left(1/9\right)\right)=2\log_{2}3.
\]
Note that it is possible to fully compute the total influence of $G_{m}$:
\[
\I\left[G_{m}\right]=2\Pr\left[G_{m}\right]=\frac{4}{3}\left(1-4^{-m}\right)
\]
and to write an expression for its spectral entropy: 
\begin{align*}
\H\left[G_{m+1}\right] & =\sum_{i=0}^{m}4^{-i}\Big[h\left(2\left(1-4^{i-m}\right)/3\right)+2h\left(\left(1-4^{i-m}\right)/3\right)\\
 & \qquad\qquad-\frac{1}{4}\tilde{h}\left(2\left(1-4^{i-m}\right)/3\right)-\frac{1}{2}\tilde{h}\left(\left(1-4^{i-m}\right)/3\right)\Big],
\end{align*}
but it is far easier to use the exponentially fast convergeance $\H\left[G_{m}\right]\xrightarrow{m\to\infty}\H\left[G\right]$,
rather than find an exact closed expression for $\H\left[G_{m}\right]$.
\begin{rem*}
Similarly, it is possible to exactly compute the total influence and
spectral entropy of $\ell\left\langle p\right\rangle $ for any rational
$p$. Indeed, every rational number has a recurrent binary representation,
yielding linear equations in $\I\left[\ell\left\langle p\right\rangle \right]$
and $\H\left[\ell\left\langle p\right\rangle \right]$. 

Approximating $\I\left[\ell\left\langle p\right\rangle \right]$ and
$\H\left[\ell\left\langle p\right\rangle \right]$ for an irrational
$p$ can be done, with exponentially decreasing error, via writing
$p$ as a limit of a sequence of dyadic rationals (e.g., truncated
binary representations of $p$).
\end{rem*}
\begin{rem*}
In a certain sense, $\ell\left\langle 2/3\right\rangle $ is the \emph{simplest}
infinite lexicographic function. Indeed, denote by $\lambda\left(p\right)$
the length of the recurring part in the binary expansion of a rational
$p$. We have $\lambda\left(p\right)=1$ if and only if~$p$ is a
dyadic rational. If~$p$ is a dyadic multiple\footnote{That is, we can write $p=a/b$ for co-prime positive integers $a$
and $b=2^{c}m$.} of $1/m$ for a positive odd integer $m$, then $\left\lceil \log_{2}m\right\rceil \le\lambda\left(p\right)\le{\rm ord}_{m}2$,
where ${\rm ord}_{m}2$ is the multiplicative order of~$2$ modulo~$m$.
In particular, $\lambda\left(p\right)\le2$ if and only if $p$ is
a dyadic multiple of $\frac{1}{3}$.
\end{rem*}
Recall that $g_{m}$ is the conjunction of two functions: $\OR_{2}\left(x_{1},x_{2}\right)$
and $G_{m}\left(x_{3},\ldots,x_{2m},x_{1}\right)$. By Fact~\ref{fact:exp-decreasing-influences},
these are almost independent since the shared variable $x_{1}$ has
exponentially small influence on $G_{m}$. 

When considering the limit object $g=\lim_{m}g_{m}$, the dependence
disappears and we have $g=\OR_{2}\sqcap G=\ell\left\langle 3/4\right\rangle \sqcap\ell\left\langle 2/3\right\rangle $,
so we can calculate its total influence and entropy using Proposition~\ref{prop:sqcap}
(full details in Appendix~\ref{apx:boring}):
\begin{eqnarray*}
I_{*} & = & \I\left[g\right]=\I\left[\OR_{2}\sqcap G\right]=\frac{2}{3}\I\left[\OR_{2}\right]+\frac{3}{4}\I\left[G\right]=\frac{2}{3}\cdot1+\frac{3}{4}\cdot\frac{4}{3}=\frac{5}{3},\\
H_{*} & = & \H\left[g\right]=\H\left[\OR_{2}\sqcap G\right]=\cdots=\frac{8}{3}+\log_{2}3,
\end{eqnarray*}
establishing our first lower bound:
\begin{thm}
\label{thm:lower-bound-1}Any constant C in Conjecture~\ref{conj:fei}
satisfies 
\[
C\ge C_{*}=H_{*}/\left(I_{*}-1\right)=4+3\log_{4}3>6.377443751,
\]
even when restricted to monotone functions.
\end{thm}

One technicality in the discussion above is that Proposition~\ref{prop:composition}
supposedly only takes a finite function, so we cannot apply it directly
to $g$, and we formally need to apply it to $g_{m}$ and let $m\to\infty$.
The slight dependence on $x_{1}$ prevents us from computing the total
influence and spectral entropy of $g_{m}$ via a direct application
of Proposition~\ref{prop:sqcap}; we can, however, consider the slight
perturbation $\tilde{g}_{m}=\OR_{2}\sqcap G_{m}$, for which Proposition~\ref{prop:sqcap}
gives $\I\left[\tilde{g}_{m}\right]\approx4/3$ and $\H\left[\tilde{g}_{m}\right]\approx\frac{8}{3}+\log_{2}3$. 

Note that $\Pr\left[\tilde{g}_{m}\right]=\frac{3}{4}\Pr\left[G_{m}\right]=\frac{1}{2}\left(1-4^{-m}\right)$,
so $\tilde{g}_{m}$ is now slightly biased, and cannot be used in
Proposition~\ref{prop:composition}. To fix that, we only need to
change a single entry of $\tilde{g}_{m}$ from $\false$ to $\true$
to get $g_{m}$ (or a different balanced function). Once again, Lemmata~\ref{lem:influence-distance}
and~\ref{lem:entropy-distance} of Section~\ref{sec:Lipschitz}
tell us that such a minuscule modification has little effect on the
entropy and total influence, which vanishes in the limit.

\section{\label{sec:LB-2}NAND on the run}

In this section we review O'Donnell and Tan's proof of Proposition~\ref{prop:composition}
and apply it, in the biased case, to the function
\[
\tau\left(x_{1},x_{2}\right)=\neg\left(x_{1}\wedge x_{2}\right).
\]

\subsection{Generalizing the composition method}

Here is a sketch of the proof of Proposition~\ref{prop:composition},
as done in~\cite[Lemma 5.1]{OT13}. A sequence of balanced Boolean
functions is built by recursively composing independent copies of~$g$.
Although both the total influence and entropy of the sequence grow
to infinity, the limit of their entropy/influence ratios is $\H\left[g\right]/\left(\I\left[g\right]-1\right)$.
For these functions to be balanced, the base function~$g$ ought
to be balanced.

The same strategy could work for a biased function~$g$ as well,
assuming its satisfies a condition that we shall immediately see.
\begin{defn*}
Fix an integer $n\ge1$. A \emph{bias} is a vector $\vec{\eta}=\left(\eta_{1},\ldots,\eta_{n}\right)$
such that $-1<\eta_{i}<1$ for $i\in\left[n\right]$. Every bias $\vec{\eta}$
induces a product measure on $\left\{ \true,\false\right\} ^{n}$
in which $\E\left[x_{i}\right]=\eta_{i}$ for $i\in\left[n\right]$
and they are pairwise independent. Denote this distribution by $x\sim\vec{\eta}$. 

Oftentimes we have $\eta_{i}=\eta$ for all $i\in\left[n\right]$,
and we denote this by $x\sim\eta$. 
\end{defn*}
\begin{example*}
The zero bias induces the uniform distribution.
\end{example*}
\begin{defn*}
A Boolean function $f$ on $n$ variables is called \emph{$\eta$-balanced}
for some $-1<\eta<1$ if $\E_{x\sim\eta}\left[f\right]=\eta$.
\end{defn*}
\begin{example*}
Balanced functions are $0$-balanced.
\end{example*}
\begin{example*}
We seek a probability $0\le p\le1$ such that $\tau$ is $\eta$-balanced
for $\eta=1-2p$, i.e.,
\[
p=\Pr_{x\sim\eta}\left[\tau\left(x\right)\right]=1-p^{2}.
\]
The polynomial $x^{2}+x-1=0$ has exactly two real roots in $\left[-1,1\right]$,
which is
\[
\Phi=\frac{\sqrt{5}-1}{2}\approx0.618034,
\]
the reciprocal of the golden ratio. Thus, $\tau$ is $\left(1-2\Phi\right)$-balanced.
\end{example*}
Two changes are required to make the proof of Proposition~\ref{prop:composition}
work when the base function~$g$ is $\eta$-balanced for $\eta\neq0$:
\begin{enumerate}
\item Computing total influence and spectral entropy under a bias. This
is provided by Lemma~\ref{lem:composition}, adapted from~\cite[Proposition 3.2]{OT13}.
\item Instead of uniform input bits, we need to start from $\eta$-biased
bits. These would be provided by lexicographic functions.
\end{enumerate}

\subsection{Biased Fourier analysis}

Let us quickly recall biased Fourier analysis of Boolean functions.
\begin{defn*}
Let $f$ be a Boolean function on~$n$ variables. For $S\subseteq\left[n\right]$,
denote by~$\tilde{\chi}_{S}$ the $\vec{\eta}$-biased basis function
\[
\tilde{\chi}_{S}=\prod_{i\in S}\frac{x_{i}-\eta_{i}}{\sqrt{1-\eta_{i}^{2}}}
\]
and denote by $\tilde{f}\left(S\right)$ the $\vec{\eta}$-biased
Fourier coefficients of $f$
\[
\tilde{f}\left(S\right)=\left\langle f,\tilde{\chi}_{S}\right\rangle =\E_{x\sim\eta}\left[f\left(x\right)\tilde{\chi}_{S}\left(x\right)\right].
\]
\end{defn*}
Since $\tilde{\chi}_{S}$ for $S\subseteq\left[n\right]$ form an
orthonormal basis of $\left\{ \true,\false\right\} ^{n}$ under the
$\vec{\eta}$-biased product measure, we still have $\sum_{S}\tilde{f}\left(S\right)^{2}=1$
and we can speak of the $\vec{\eta}$-biased spectral distribution
$\tilde{p}_{f}\left(S\right)=\tilde{f}\left(S\right)^{2}$ of $f$,
and consequently, the $\vec{\eta}$-biased total influence $\tilde{\I}\left[f\right]$
and $\vec{\eta}$-biased spectral entropy $\tilde{\H}\left[f\right]$.
\begin{example*}
Let $\eta=1-2\Phi=-\Phi^{3}$. Given that $\sqrt{1-\eta^{2}}=2\sqrt{\Phi\left(1-\Phi\right)}=2\Phi^{3/2}$,
the $\eta$-biased spectral distribution of $\tau$ is: 
\[
\tilde{p}_{\tau}\left(S\right)=\begin{cases}
\Phi^{6}, & S=\varnothing;\\
4\Phi^{5}, & \left|S\right|=1;\\
4\Phi^{6}, & \left|S\right|=2,
\end{cases}
\]
so its $\eta$-biased total influence and spectral entropy are (full
details in Appendix~\ref{apx:boring}):
\begin{eqnarray*}
\tilde{\I}\left[\tau\right] & = & \Phi^{6}\cdot0+2\cdot4\Phi^{5}\cdot1+4\Phi^{6}\cdot2=8\left(\Phi^{5}+\Phi^{6}\right)=8\Phi^{4}\approx1.16718,\\
\tilde{\H}\left[\tau\right] & = & \cdots\;\;\,=8\left(1-2\Phi\right)+10\left(4\Phi-3\right)\log_{2}\Phi\approx1.77611.
\end{eqnarray*}
\end{example*}

\subsection{Composition lemma}

To simplify the notation of Lemma~\ref{lem:composition}, we introduce
a variant of the total influence and entropy definitions.
\begin{defn*}
Let $f$ be a Boolean function and let $S\sim p_{f}$. The unbiased
total influence and unbiased entropy of $f$, denoted respectively
by $\I^{+}\left[f\right]$ and $\H^{+}\left[f\right]$, are
\begin{eqnarray*}
\I^{+}\left[f\right] & = & \E\left[\left|S\right|\mid S\neq\varnothing\right]=\frac{\I\left[f\right]}{\V\left[f\right]};\\
\H^{+}\left[f\right] & = & \H\left[S\mid S\neq\varnothing\right]=\frac{\H\left[f\right]-h(\V\left[f\right])}{\V\left[f\right]}=\frac{\H\left[f\right]-\tilde{h}\left(\Pr\left[f\right]\right)}{\V\left[f\right]},
\end{eqnarray*}
where $\V\left[f\right]=\Pr\left[S\neq\varnothing\right]=1-\E\left[f\right]^{2}=4\Pr\left[f\right]\left(1-\Pr\left[f\right]\right)$.
\end{defn*}
\begin{example*}
For $f=\AND_{n}$ we have
\begin{eqnarray*}
\I^{+}\left[\AND_{n}\right] & = & \frac{\I\left[\AND_{n}\right]}{\V\left[\AND_{n}\right]}=\frac{2n/N}{4/N\left(1-1/N\right)}=\frac{n}{2\left(1-1/N\right)};\\
\H^{+}\left[\AND_{n}\right] & = & \H\left[S\mid S\neq\varnothing\right]=\H\left[{\rm Uniform}\left(N-1\right)\right]=\log_{2}\left(N-1\right)\approx n.
\end{eqnarray*}
\end{example*}
\begin{lem}[{\cite[Proposition 3.2]{OT13}}]
\label{lem:composition}Let $F$ be a Boolean function on $k$ variables
and let $g_{1},\ldots,g_{k}$ be Boolean functions on $n$ variables.
Define a Boolean function $f=F\circ\left(g_{1},\ldots,g_{k}\right)$
on $kn$ variables by
\[
f\left(x_{1},\ldots,x_{kn}\right)=F\left(g_{1}\left(x_{1},\ldots,x_{n}\right),g_{2}\left(x_{n+1},\ldots,x_{2n}\right),\ldots,g_{k}\left(x_{\left(k-1\right)n+1},\ldots,x_{kn}\right)\right).
\]
Then
\begin{eqnarray*}
\I\left[f\right] & = & \sum_{i=1}^{k}\tilde{\I_{i}}\left[F\right]\I^{+}\left[g_{i}\right];\\
\H\left[f\right] & = & \sum_{i=1}^{k}\tilde{\I_{i}}\left[F\right]\H^{+}\left[g_{i}\right]+\tilde{\H}\left[F\right],
\end{eqnarray*}
where $\tilde{p}_{F}$ is the $\vec{\eta}$-biased spectral distribution
of $F$ for the bias $\vec{\eta}=\left(\E\left[g_{1}\right],\ldots,\E\left[G_{k}\right]\right)$
and $\tilde{\I}_{i}\left[F\right]=\Pr_{S\sim\tilde{p}_{F}}\left[i\in S\right]$
for $i\in\left[k\right]$. In particular, when $g_{i}=g$ for all
$i\in\left[k\right]$ we get
\begin{eqnarray*}
\I\left[f\right] & = & \tilde{\I}\left[F\right]\I^{+}\left[g\right];\\
\H\left[f\right] & = & \tilde{\I}\left[F\right]\H^{+}\left[g\right]+\tilde{\H}\left[F\right].
\end{eqnarray*}
\end{lem}

\subsection{A second lower bound}

Define a sequence of functions $\left(F_{m}\right)_{m\ge0}$ by $F_{0}=\ell\left\langle \Phi\right\rangle $
and $F_{m+1}=\tau\circ\left(F_{m},F_{m},F_{m}\right)$ for all $m\ge0$.
Recall that $\tau$ is $\left(1-2\Phi\right)$-balanced and thus $\Pr\left[F_{m}\right]=\Phi$
for all $m\ge0$. 

Via Lemma~\ref{lem:composition} we can compute the asymptotic entropy/influence
ratio of $F_{m}$ (see Appendix~\ref{apx:boring} for details):
\begin{claim}
\label{clm:Fm-ratio}${\displaystyle \lim_{m\to\infty}\frac{\H\left[F_{m}\right]}{\I\left[F_{m}\right]}=\frac{\H\left[\ell\left\langle \Phi\right\rangle \right]+\left(3+2\Phi\right)\tilde{\H}\left[\tau\right]-\left(4+2\Phi\right)\tilde{h}\left(\Phi\right)}{\I\left[\ell\left\langle \Phi\right\rangle \right]}.}$
\end{claim}

\begin{thm}
\label{thm:lower-bound-2}Any constant C in Conjecture~\ref{conj:fei}
satisfies $C>6.413846$.
\end{thm}

\begin{proof}
Plug in Claim~\ref{clm:Fm-ratio} the value 
\[
\tilde{\H}\left[\tau\right]=8\left(1-2\Phi\right)+10\left(4\Phi-3\right)\log_{2}\Phi>1.7761
\]
computed above, and the approximations $\I\left[\ell\left\langle \Phi\right\rangle \right]<1.2976895$,
$\H\left[\ell\left\langle \Phi\right\rangle \right]>2.4239395$.
\end{proof}
\begin{rem*}
Using any $\left(1-2p\right)$-balanced base function $g$ for $0<p<1$,
the same computation yields the lower bound:
\[
C\ge\frac{\H\left[\ell\left\langle p\right\rangle \right]}{\I\left[\ell\left\langle p\right\rangle \right]}+\frac{4p\left(1-p\right)\tilde{\H}\left[g\right]-\tilde{\I}\left[g\right]\tilde{h}\left(p\right)}{\I\left[\ell\left\langle p\right\rangle \right]\left(\tilde{\I}\left[g\right]-4p\left(1-p\right)\right).}
\]
If $g$ is balanced, i.e., $p=\frac{1}{2}$, then this is plainly
$\H\left[g\right]/\left(\I\left[g\right]-1\right)$, recovering Proposition~\ref{prop:composition}.
\end{rem*}

\section{\label{sec:LB-3}To Infinity, and Beyond}

In both Theorem~\ref{thm:lower-bound-1} and Theorem~\ref{thm:lower-bound-2}
the notion of limit was used twice: 
\begin{enumerate}
\item In creating an infinite lexicographic function ($\ell\left\langle 2/3\right\rangle $
and $\ell\left\langle \Phi\right\rangle $, respectively); and
\item When taking the asymptotic entropy/influence ratio for the sequence
of functions defined by recursive composition.
\end{enumerate}
In this section we use limits a \emph{countable} number times for
a superior construction. 

\subsection{Limit of limits}

The basic step is inspired by the NAND function $\tau$ of Section~\ref{sec:LB-2}.

Fix a Boolean function~$\lambda$, and define a function~$\kappa$
using the equation $\kappa=\left(\lambda\sqcap\kappa\right)^{\dagger}$.
Formally, we define a sequence $\left(\kappa_{m}\right)_{m\ge0}$
of functions via $\kappa_{0}=\lambda$ and $\kappa_{m+1}=\left(\lambda\sqcap\kappa_{m}\right)^{\dagger}$
and let $\kappa=\lim_{m\to\infty}\kappa_{m}$. 
\begin{prop}
\label{prop:step}Write $p=\Pr\left[\lambda\right]$ and $q=\Pr\left[\kappa\right]$.
Then $\kappa$ satisfies $q=1/\left(1+p\right)$ and
\begin{eqnarray*}
\I^{+}\left[\kappa\right] & = & \I^{+}\left[\lambda\right]/q\\
\H^{+}\left[\kappa\right] & = & \left(\H^{+}\left[\lambda\right]+\frac{h\left(p\right)}{1-p}\right)/q.
\end{eqnarray*}
\end{prop}

\begin{proof}
Apply Proposition~\ref{prop:sqcap} (see Appendix~\ref{apx:boring}
for the computation).
\end{proof}
\begin{rem*}
Note that when $\lambda$ is monotone, $\kappa$ is monotone as well.
\end{rem*}
Fix some Boolean function $F_{0}$ and let $z=\Pr\left[F_{0}\right]$,
$I=\I\left[F_{0}\right]$ and $H=\H\left[F_{0}\right]$. Define a
sequence $\left(F_{m}\right)_{m\ge1}$ using the equation $F_{m+1}=\left(F_{m}\sqcap F_{m+1}\right)^{\dagger}$,
and let $q_{m}=\Pr\left[F_{m}\right]$. By Proposition~\ref{prop:step},\begin{subequations}
\begin{eqnarray}
q_{m+1} & = & \frac{1}{1+q_{m}};\label{eq:q-sequence}\\
\I^{+}\left[F_{m+1}\right] & = & \I^{+}\left[F_{m}\right]/q_{m+1};\label{eq:I-sequence}\\
\H^{+}\left[F_{m+1}\right] & = & \left(\H^{+}\left[F_{m}\right]+\frac{h\left(q_{m}\right)}{1-q_{m}}\right)/q_{m+1}.\label{eq:H-sequence}
\end{eqnarray}
\end{subequations}These values naturally depend on the initial parameters
$\left(z,I,H\right)$. Nevertheless, the sequence $q_{m}$ converges
to the same limit $\lim_{m}q_{m}=\Phi$ regardless of the choice of
$0\le z\le1$.

In fact, $q_{m}$ is a linear rational function of $z$, as the following
claim states.
\begin{claim}
\label{clm:qm-linear-rational}For all $m\ge0$ we have 
\begin{equation}
q_{m}=\frac{b_{m-1}z+b_{m}}{b_{m}z+b_{m+1}},\label{eq:explicit-q}
\end{equation}
where 
\begin{equation}
b_{m}=\frac{\Phi^{-m}-\left(-\Phi\right)^{m}}{\sqrt{5}}\label{eq:binet-formula}
\end{equation}
is the $m$th Fibonacci number.
\end{claim}

\begin{proof}
By induction on $m$ (see Appendix~\ref{apx:boring} for details).
\end{proof}
\begin{rem*}
Binet's formula~(\ref{eq:binet-formula}) naturally extends the Fibonacci
sequence to $\mathbb{Z}$. Note that for all $m\in\mathbb{Z}$ we
have $b_{-m}=\left(-1\right)^{m+1}b_{m}$.

This can be used to define $q_{m}$ for negative $m$ as well. Note
that for $m<0$ we are no longer promised that $0\le q_{m}\le1$.
In particular, $q_{-k}$ is undefined for $z=b_{k}/b_{k+1}$. 
\end{rem*}
\begin{cor*}
For all $n\ge0$ we have 
\begin{equation}
\pi_{m}=\frac{1}{zb_{m}+b_{m+1}},\label{eq:explicit-pi}
\end{equation}
where $\pi_{m}$ is the cumulative product $\pi_{m}=\prod_{k=1}^{m}q_{k}$.
In particular, 
\begin{equation}
\frac{1}{2}\Phi^{m-1}\le\frac{1}{b_{m+2}}=\frac{1}{b_{m}+b_{m+1}}\le\pi_{m}\le\frac{1}{b_{m+1}}\le\Phi^{m-1}.\label{eq:pi-bounds}
\end{equation}
\end{cor*}
For notational convenience, write $\pi_{-1}=\pi_{0}/q_{0}=1/z$ and
$\pi_{-2}=\pi_{-1}/q_{-1}=1/\left(1-z\right)$. The next claim computes
the entropy/influence ratio of $F_{m}$ via~(\ref{eq:I-sequence})
and~(\ref{eq:H-sequence}):
\begin{claim}
\label{clm:Fm-ratio-redux}For all $m\ge0$ we have 
\[
\frac{\H\left[F_{m}\right]}{\I\left[F_{m}\right]}=\frac{1}{I}\left(H-\tilde{h}\left(z\right)+\beta_{m}\left(z\right)+z\left(1-z\right)\left(\pi_{m-1}+\pi_{m-2}\right)\tilde{h}\left(q_{m}\right)\right),
\]
where $\beta_{m}:\left[0,1\right]\to\mathbb{R}$ is 
\begin{align*}
\beta_{m}\left(z\right) & =4z\left(1-z\right)\sum_{k=-2}^{m-3}h\left(q_{k+2}\right)\pi_{k}\\
 & =\begin{cases}
0, & m=0;\\
4zh\left(z\right), & m=1;\\
4zh\left(z\right)+4\left(1-z\right)h\left(q_{1}\right)+4z\left(1-z\right)\sum_{k=0}^{m-3}h\left(q_{k+2}\right)\pi_{k}, & m\ge2,
\end{cases}
\end{align*}
\end{claim}

\begin{proof}
See Appendix~\ref{apx:boring}.
\end{proof}
\begin{rem*}
Using~(\ref{eq:explicit-q}) and~(\ref{eq:explicit-pi}) we can
write $\beta_{m}$ as an explicit function of $z$ for all $m\in\mathbb{N}$.
\end{rem*}
\medskip{}
Asymptotically the term $\left(\pi_{m-1}+\pi_{m-2}\right)\tilde{h}\left(q_{m}\right)$
vanishes, and we obtain 
\begin{equation}
\lim_{m\to\infty}\frac{\H\left[F_{m}\right]}{\I\left[F_{m}\right]}=\frac{H-\tilde{h}\left(z\right)+\beta\left(z\right)}{I},\label{eq:Fm-ratio}
\end{equation}
where 
\begin{align*}
\beta\left(z\right) & =\lim_{m\to\infty}\beta_{m}\left(z\right)=4z\left(1-z\right)\sum_{k=-2}^{\infty}h\left(q_{k+2}\right)\pi_{k}\\
 & =4zh\left(z\right)+4\left(1-z\right)h\left(q_{1}\right)+4z\left(1-z\right)\sum_{k=0}^{\infty}h\left(q_{k+2}\right)\pi_{k}.
\end{align*}
By~(\ref{eq:pi-bounds}), $\beta_{m}\xrightarrow{m\to\infty}\beta$
exponentially fast to $\beta$, and thus $\H\left[F_{m}\right]/\I\left[F_{m}\right]$
converges very quickly as well. This can be seen visually in Figure~\ref{fig:beta}.
\begin{thm}
\label{thm:lower-bound-3}Any constant C in Conjecture~\ref{conj:fei}
satisfies 
\[
C\ge\beta\left(1/2\right)>6.4547837,
\]
even when restricted to monotone functions.
\end{thm}

\begin{proof}
Select $F_{0}=\iota$ with parameters $z=1/2$, $I=1$ and $H=0$.
Alternatively, since for $F_{0}=\iota$ we have $F_{1}=\ell\left\langle 2/3\right\rangle $,
we could start with $F_{0}=\ell\left\langle 2/3\right\rangle $ and
get the same bound.
\end{proof}

\subsection{Afterthoughts}
\begin{enumerate}
\item One may ask herself whether it would suffice to define a limit function
$T$ using the equation $T=\left(T\sqcap T\right)^{\dagger}$. This
is equivalent to the composition construction of Theorem~\ref{thm:lower-bound-2},
but we get a monotone function and are no longer limited to using
$T_{0}=\ell\left\langle \Phi\right\rangle $. It is possible to show,
via a computation quite similar to the one in previous subsection
(see Appendix~\ref{apx:boring}), that 
\[
\lim_{m\to\infty}\frac{\H\left[T_{m}\right]}{\I\left[T_{m}\right]}=\frac{H-\tilde{h}\left(z\right)+\gamma\left(z\right)}{I}
\]
for a function $\gamma$ slightly smaller than $\beta$. Picking $T_{0}=\ell\left\langle \Phi\right\rangle $
is actually quite far from being optimal here; $T_{0}=\iota$ or $T_{0}=\ell\left\langle 5/8\right\rangle $
yield a lower bound of $\approx6.44539$, while $T_{0}=\ell\left\langle 2/3\right\rangle $
seems to attain the best lower bound $\approx6.453111$ achievable
using this method. This comes close, but is still less than $\beta\left(1/2\right)$,
since $\gamma\left(2/3\right)<\beta\left(2/3\right)$.
\item Recall that the decision tree complexity of an infinite lexicographic
function is just two bits. By simple induction, it can be shown that
the average decision tree complexity of~$F_{m}$ is $2^{m}d$ for
all $m\ge0$, where~$d$ is the average decision tree complexity
of~$F_{0}$.

In particular, the average decision tree complexity of the sequence
$\left(F_{m}\right)_{m\ge0}$ is unbounded, and thus the construction
is not subject to the upper bound on constant average depth decision
trees of~\cite{WWW14}. Each $F_{m}$ is still computable by a read-once
formula, though.
\item The half circle shape of $\beta\left(z\right)=4z\left(1-z\right)\sum_{k=-2}^{\infty}h\left(q_{k+2}\right)\pi_{k}$
is mostly dictated by the variance term $4z\left(1-z\right)$, which
is symmetric about $z=\frac{1}{2}$. One may thus guess that $\max_{z}\beta\left(z\right)=\beta\left(1/2\right)$.
Surprisingly, the maximum of $\beta$ is obtained at $z^{*}\approx0.50168825$,
giving a meager improvement of 0.006\% over $\beta\left(1/2\right)$. 

Nevertheless, it seems this cannot be used to improve the bound of
Theorem~\ref{thm:lower-bound-3}. Indeed, any change in $z$ will
have a negative effect on~(\ref{eq:Fm-ratio}) by increasing both
$I$ and $\tilde{h}\left(z\right)$, so to gain anything we need the
initial function $F_{0}$ to provide a large entropy/influence ratio,
which is what we were seeking all along.

Furthermore, any \emph{balanced} function $F_{0}$ beating $\beta\left(1/2\right)$
must have $H>\left(I-1\right)\beta\left(1/2\right)$, so we could
have used it in Proposition~\ref{prop:composition} directly!
\begin{figure}[H]
\begin{centering}
\includegraphics[width=1\textwidth]{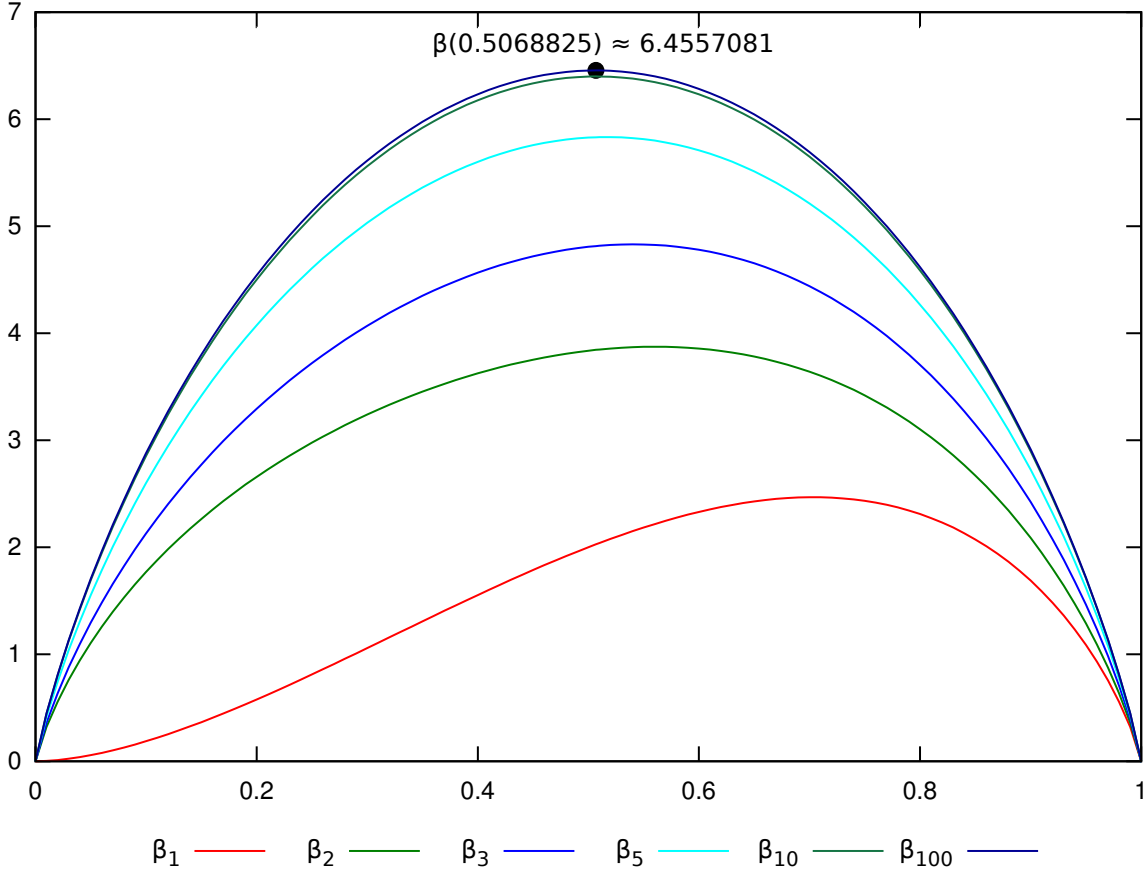}
\par\end{centering}
\caption{\label{fig:beta}The functions $\beta_{m}$ for $m=1,2,3,5,10,100$.}
\end{figure}
\item The function $\beta$ can be simplified a bit further. Observe that
\begin{align*}
h\left(q_{m+1}\right) & =h\left(\frac{1}{1+q_{m}}\right)=\frac{1}{1+q_{m}}\log_{2}\left(1+q_{m}\right)+\frac{q_{m}}{1+q_{m}}\log_{2}\frac{1+q_{m}}{q_{m}}\\
 & =\log_{2}\left(1+q_{m}\right)-q_{m}q_{m+1}\log_{2}q_{m}=-\log_{2}q_{m+1}-q_{m}q_{m+1}\log_{2}q_{m},
\end{align*}
hence we can write 
\[
\beta\left(z\right)=4zh\left(z\right)-4z\left(1-z\right)\sum_{m=0}^{\infty}\left(\pi_{m}\log_{2}q_{m+1}+\pi_{m+2}\log_{2}q_{m}\right).
\]
\end{enumerate}

\section{\label{sec:Lipschitz}A Lipschitz-type condition for total influence
and entropy}

In this section we show that changing a single entry in a Boolean
function has a negligible effect on its total influence and entropy.
\begin{lem}
\label{lem:influence-distance}Let $f$ and $g$ be Boolean functions
on $n$ variables differing in a single entry $x=x_{0}$. Then $\left|\I\left[f\right]-\I\left[g\right]\right|\le2n/N$. 
\end{lem}

\begin{proof}
We use an equivalent definition of total influence as the average
sensitivity 
\[
\I\left[f\right]=\E\left[\S_{f}\left(x\right)\right]=\frac{1}{N}\sum_{x}\S_{f}\left(x\right),
\]
where $\S_{f}\left(x\right)$ is the number of neighbors $y\sim x$
in the Boolean cube such that $f\left(x\right)\neq f\left(y\right)$.
Indeed, we have 
\[
\left|\S_{f}\left(x\right)-\S_{g}\left(x\right)\right|\le\begin{cases}
n, & x=x_{0};\\
1, & x\sim x_{0};\\
0, & \text{otherwise}.
\end{cases}
\]
Thus,
\[
\left|\I\left[f\right]-\I\left[g\right]\right|=\frac{1}{N}\left|\sum_{x}\S_{f}\left(x\right)-\S_{g}\left(x\right)\right|\le\frac{1}{N}\sum_{x}\left|\S_{f}\left(x\right)-\S_{g}\left(x\right)\right|\le\frac{2n}{N}.\qedhere
\]
\end{proof}
\begin{rem*}
This is tight. Indeed, ${\rm Or}_{n}$ differs from the all-$\true$
function in a single entry $x=\false^{n}$ and $\I\left[\OR_{n}\right]-\I\left[\true\right]=2n/N-0=2n/N$.
\end{rem*}
\begin{lem}
\label{lem:entropy-distance}Let $f$ and $g$ be Boolean functions
on $n$ variables differing in a single entry $x=x_{0}$. Then $\left|\H\left[f\right]-\H\left[g\right]\right|\le12n/\sqrt{N}$.
\end{lem}

\begin{proof}
This is trivial for $n=1$ so assume $n\ge2$. Also assume without
loss of generality that the differing entry is $x_{0}=\false^{n}$;
that is, $g=f+2\delta$, where 
\[
\delta\left(x\right)=\frac{1}{2}+\frac{1}{2}\OR_{n}\left(x\right)=\begin{cases}
1, & x=x_{0};\\
0, & x\neq x_{0}.
\end{cases}
\]
Write $a=\frac{N}{2}\left(f+g\right)=N\left(f+\delta\right)$, so
$f=\frac{1}{N}a-\delta$ and thus $\hat{f}\left(S\right)=\frac{1}{N}\left(\hat{a}\left(S\right)-1\right)$.
Similarly we have $g=\frac{1}{N}a+\delta$ and $\hat{g}\left(S\right)=\frac{1}{N}\left(\hat{a}\left(S\right)+1\right)$.
In particular, we have 
\begin{eqnarray*}
\sum_{S\subseteq\left[n\right]}\hat{a}\left(S\right) & = & N\left(\sum_{S\subseteq\left[n\right]}\hat{f}\left(S\right)\right)-N=0;\\
\sum_{S\subseteq\left[n\right]}\hat{a}\left(S\right)^{2} & = & \sum_{S\subseteq\left[n\right]}\left(N^{2}\hat{f}\left(S\right)^{2}+2N\hat{f}\left(S\right)+1\right)=N\left(N-1\right).
\end{eqnarray*}
Fourier coefficients of Boolean functions on $n$ variables are known
to reside in $\left\{ \pm2k/N\right\} _{k=0}^{N/2}$; thus the Fourier
coefficients of $a$ belong to $\left\{ \pm\left(2k-1\right)\right\} _{k=1}^{N/2}$.
For $k\in\left[N/2\right]$ let 
\[
\Delta_{k}=\sum_{S:\left|\hat{a}\left(S\right)\right|=2k-1}{\rm sgn}\left(\hat{a}\left(S\right)\right)=\left|\left\{ S:\hat{a}\left(S\right)=2k-1\right\} \right|-\left|\left\{ S:\hat{a}\left(S\right)=1-2k\right\} \right|
\]
and observe that \begin{subequations}
\begin{eqnarray}
\sum_{k=1}^{N/2}\Delta_{k}\left(2k-1\right) & = & \sum_{S\subseteq\left[n\right]}\hat{a}\left(S\right)=0,\label{eq:delta-sum}\\
\sum_{k=1}^{N/2}\left|\Delta_{k}\right| & \le & \sum_{S\subseteq\left[n\right]}1=N,\label{eq:delta-sum-abs0}\\
\sum_{k=1}^{N/2}\left|\Delta_{k}\right|\left(2k-1\right)^{2} & \le & \sum_{S\subseteq\left[n\right]}\hat{a}\left(S\right)^{2}=N\left(N-1\right),\label{eq:delta-sum-abs2}\\
\sum_{k=1}^{N/2}\left|\Delta_{k}\right|\left(2k-1\right) & \le & \sqrt{\sum_{k=1}^{N/2}\left|\Delta_{k}\right|\cdot\sum_{k=1}^{N/2}\left|\Delta_{k}\right|\left(2k-1\right)^{2}}\le N\sqrt{N-1}<N^{3/2},\label{eq:delta-sum-abs1}
\end{eqnarray}
\end{subequations}where Cauchy\textendash Schwartz was used to derive~(\ref{eq:delta-sum-abs1})
from~(\ref{eq:delta-sum-abs0}) and~(\ref{eq:delta-sum-abs2}). 

\medskip{}

We express the difference of entropies in terms of $\left(\Delta_{k}\right)_{k\in\left[N/2\right]}$
(details in Appendix~\ref{apx:boring}):
\begin{equation}
\H\left[f\right]-\H\left[g\right]=\frac{8}{N^{2}}\sum_{k=2}^{N/2}\Delta_{k}k^{2}\log_{2}\frac{k}{k-1}+\frac{8}{N^{2}}\sum_{k=2}^{N/2}\Delta_{k}\left(2k-1\right)\log_{2}\left(k-1\right).\label{eq:entropy-diff}
\end{equation}
To bound the first term, note that the function $\xi\left(x\right)=x^{2}\log_{2}\left(x/\left(1-x\right)\right)-\left(2x-1\right)/\ln4$
is decreasing and positive for $x>1$, and $\xi\left(2\right)=4-3/\ln4\approx1.836$.
Now
\begin{align}
\Big|\sum_{k=2}^{N/2}\Delta_{k}k^{2}\log_{2}\frac{k}{k-1} & \Big|=\Big|\sum_{k=2}^{N/2}\Delta_{k}\left(\xi(k)+\frac{2k-1}{\ln4}\right)\Big|\nonumber \\
 & \le\frac{1}{\ln4}\Big|\sum_{k=2}^{N/2}\Delta_{k}\left(2k-1\right)\Big|+\Big|\sum_{k=2}^{N/2}\Delta_{k}\xi(k)\Big|\nonumber \\
\text{[by\,\eqref{eq:delta-sum}]} & =\frac{1}{\ln4}\left|\Delta_{1}\right|+\Big|\sum_{k=2}^{N/2}\Delta_{k}\xi(k)\Big|\nonumber \\
 & \le\frac{1}{\ln4}\left|\Delta_{1}\right|+\sum_{k=2}^{N/2}\left|\Delta_{k}\right|\xi(2)\nonumber \\
\text{} & \le\max\left\{ \frac{1}{\ln4},\xi(2)\right\} \sum_{k=1}^{N/2}\left|\Delta_{k}\right|\nonumber \\
\text{[by\,\eqref{eq:delta-sum-abs0}]} & \le\max\left\{ \frac{1}{\ln4},\xi(2)\right\} N=\xi(2)N<2N.\label{eq:term1}
\end{align}
To bound the second term, note that the function $\zeta\left(x\right)=x\log_{2}\frac{x-1}{2}$
is increasing, positive and convex for $x>3$, so
\begin{align}
\Big|2\sum_{k=2}^{N/2}\Delta_{k}\left(2k-1\right)\log_{2}\left(k-1\right)\Big| & =2\Big|\sum_{k=2}^{N/2}\Delta_{k}\zeta\left(2k-1\right)\Big|\le2\sum_{k=2}^{N/2}\left|\Delta_{k}\right|\zeta\left(2k-1\right)\nonumber \\
\text{[Jensen's inequality]} & \le2\zeta\Big(\sum_{k=2}^{N/2}\left|\Delta_{k}\right|\left(2k-1\right)\Big)\nonumber \\
\text{[by\,\eqref{eq:delta-sum-abs1}]} & \le2\zeta\left(N^{3/2}\right)=2N^{3/2}\log_{2}\frac{N^{3/2}-1}{2}\nonumber \\
 & <2N^{3/2}\left(3n/2-1\right)=3N^{3/2}n-2N^{3/2}.\label{eq:term2}
\end{align}
Combining~(\ref{eq:entropy-diff}), (\ref{eq:term1}) and~(\ref{eq:term2}),
we get 
\[
\left|\H\left[f\right]-\H\left[g\right]\right|\le\frac{8}{N^{2}}\cdot2N+\frac{4}{N^{2}}\left(3N^{3/2}n-2N^{3/2}\right)\le\frac{12n}{N},
\]
establishing the proof of the lemma. 
\end{proof}
\begin{cor*}
Let $f$ and $g$ be Boolean functions on $n$ variables, and let
$\epsilon=\Pr\left[f\left(x\right)\neq g\left(x\right)\right]$. Then
$\left|\I\left[f\right]-\I\left[g\right]\right|\le2\epsilon n$ and
$\left|\H\left[f\right]-\H\left[g\right]\right|<12\epsilon n\sqrt{N}$.
\end{cor*}
\begin{rem*}
One may wonder how tight is Lemma~\ref{lem:entropy-distance}, since
the largest distance obtained from natural examples seems to be 
\[
\H\left[\OR_{n}\right]-\H\left[\true\right]<8\left(n-1+1/\ln4\right)/N-0<8n/N.
\]
There exists, however, a algebraic construction in which the entropy
difference is greater than $8/(3\sqrt{N})$, so the lemma is tight
upto the logarithmic factor $n$:

Niho~\cite{Niho72} considered functions from $\mathbb{GF}\left(N\right)$
to $\mathbb{GF}\left(2\right)$ of the form $f\left(\alpha\right)={\rm Tr}\left(\alpha^{r}\right)$,
where ${\rm Tr}:\mathbb{GF}\left(N\right)\to\mathbb{GF}\left(2\right)$
is the trace operator and $r$ is some integer. These can naturally
be interpreted as Boolean functions from $\left\{ \false,\true\right\} ^{n}$
to $\left\{ \false,\true\right\} $.

The case when $n\equiv0\bmod4$ and $r=2\sqrt{N}-1$ was analyzed
in~\cite[Theorem 3-6]{Niho72}. The Fourier spectrum of the resulting
function $f$ has four possible values, as summarized in Table~\ref{tab:niho-spectrum}.
Plugging these numbers in~(\ref{eq:entropy-diff}) shows that indeed
$\H\left[f+2\delta\right]-\H\left[f\right]>8/(3\sqrt{N})$.

\begin{table}[H]
\begin{centering}
\begin{tabular}{|c|c|c|c|c|}
\hline 
Value of $\hat{f}$ & $-1/\sqrt{N}$ & $0$ & $1/\sqrt{N}$ & $2/\sqrt{N}$\tabularnewline
\hline 
\hline 
Multiplicity & $\frac{1}{3}(N-\sqrt{N})$ & $\frac{1}{2}(N-\sqrt{N})$ & $\sqrt{N}$ & $\frac{1}{6}(N-\sqrt{N})$\tabularnewline
\hline 
\end{tabular}
\par\end{centering}
\caption{\label{tab:niho-spectrum}Spectrum of Niho's function $f\left(\alpha\right)={\rm Tr}\big(\alpha^{2\sqrt{N}-1}\big)$
for $n\equiv0\bmod4$.}

\end{table}

\end{rem*}

\section{Concluding Remarks and Open Problems}
\begin{enumerate}
\item The key element repeating in all our constructions is lexicographic
functions:
\begin{enumerate}
\item In Theorem~\ref{thm:lower-bound-1} they allowed us to create a balanced
function $\ell\left\langle 2/3\right\rangle \sqcap\ell\left\langle 3/4\right\rangle $
of positive entropy and small total influence that we could plug in
Proposition~\ref{prop:composition}; 
\item In Theorem~\ref{thm:lower-bound-2}, using $\ell\left\langle \Phi\right\rangle $
we converted uniform bits to $\left(1-2\Phi\right)$-biased bits that
we could plug in the biased variant of Proposition~\ref{prop:composition}
with the base function~$\tau$; 
\item In Theorem~\ref{thm:lower-bound-3}, the constructed sequence had
$\ell\left\langle 2/3\right\rangle $ as either its first or second
member.
\end{enumerate}
This is far from being a coincidence, as lexicographic functions are
the minimizers of total influence (for a given bias) by Theorem~\ref{thm:edge-isoperimetric-ineq-harper}.
It seems plausible to attempt proving Conjecture~\ref{conj:fei}
for the class of monotone Boolean functions (or perhaps all Boolean
functions) by proving an upper bound on what can be done using lexicographic
functions.
\item It is possible that we can improve on Theorems~\ref{thm:lower-bound-2}
and~\ref{thm:lower-bound-3} by finding a base function~$g$ better
than~$\tau$. Of course, $g$ should be $\eta$-balanced for some
$-1<\eta<1$; that is, $\eta$ should be a fixed point of $E_{g}\left(\rho\right)=\E_{x\sim\rho}\left[g\right]$.
Preferably, $\eta$ should be an \emph{attractive} fixed point of
$E_{g}$, so~$g$ needs to be a non-monotone function. By exhaustive
search, we have determined that no function on $n\le4$ variables
will do better than $\tau$.

Nevertheless, all constructions based on disjoint composition belong
to the class of read-once formulas, and thus cannot provide a lower
bound better than 10.
\item One remaining gap worth closing is the asympotic behavior of the Lipschitz
constant for the spectral entropy. Recall that Niho's function gave
a lower bound of $\Omega(1/\sqrt{N})$, whereas the upper bound provided
by Lemma~\ref{lem:entropy-distance} is $O(n/\sqrt{N})$. We believe
the upper bound is not tight.
\end{enumerate}

\section*{Acknowledgements}

The author wishes to thank Nathan Keller for fruitful discussions
and suggestions, and Ohad Klein for useful comments.

\appendix

\section{\label{apx:boring}Boring Calculations}

\subsection{From Section~\ref{sec:intro}}
\begin{proof}[Proof of Claim~\ref{clm:Gm-2-3-conditioned}]
By induction on $m$. It is true for $m=1$ since 
\[
\Pr\left[x_{1}\mid x_{1}\vee x_{2}\right]=\Pr\left[x_{1}\wedge\left(x_{1}\vee x_{2}\right)\right]/\Pr\left[x_{1}\vee x_{2}\right]=\frac{1/2}{3/4}=2/3.
\]
Assuming correctness for $m$, we have
\begin{align*}
 & \Pr\left[G_{m+1}\left(x_{3},\ldots,x_{2m+2},x_{1}\right)\mid x_{1}\vee x_{2}\right]\\
 & \quad=\Pr\left[x_{3}\vee\left(x_{4}\wedge G_{m}\left(x_{5},\ldots,x_{2m+2},x_{1}\right)\right)\mid x_{1}\vee x_{2}\right]\\
 & \quad=1-\left(1-\frac{1}{2}\right)\left(1-\frac{1}{2}\cdot\Pr\left[G_{m}\left(x_{5},\ldots,x_{2m+2},x_{1}\right)\mid x_{1}\vee x_{2}\right]\right)\\
 & \quad=1-\frac{1}{2}\left(1-\frac{1}{2}\cdot\frac{2}{3}\right)=\frac{1}{2}+\frac{1}{6}=\frac{2}{3}.\qedhere
\end{align*}
\end{proof}

\subsection{From Section~\ref{sec:LB-1}}
\begin{proof}[Proof of Claim~\ref{clm:lex-influence}]
By induction on $t$. It trivially holds for $t=0$. Assuming correctness
for~$t$, let $s'=s-N/2^{k_{0}}$. By Proposition~\ref{prop:lex-influence}
\begin{align*}
\I\left[\ell_{n}\left\langle s\right\rangle \right] & =\frac{2sn}{N}-\frac{4}{N}\sum_{x=0}^{s-1}wt\left(x\right)\\
 & =\frac{2n}{N}\cdot\frac{N}{2^{k_{0}}}+\frac{2s'n}{N}-\frac{4}{N}\sum_{x=0}^{N/2^{k_{0}}-1}wt\left(x\right)-\frac{4}{N}\sum_{x=N/2^{k_{0}}}^{s-1}wt\left(x\right)\\
 & =\frac{n}{2^{k_{0}-1}}-\frac{4}{N}\frac{N}{2^{k_{0}}}\frac{n-k_{0}}{2}+\frac{2s'n}{N}-\frac{4}{N}\sum_{x=0}^{s'-1}\left(1+wt\left(x\right)\right)\\
 & =\frac{n}{2^{k_{0}-1}}-\frac{n-k_{0}}{2^{k_{0}-1}}+\I\left[\ell_{n}\left\langle s'\right\rangle \right]-\frac{4s'}{N}\\
 & =k_{0}2^{1-k_{0}}+\sum_{i=1}^{t+1}\left(k_{i}-2i+2\right)2^{1-k_{i}}-\frac{4}{N}\sum_{i=1}^{t+1}N/2^{k_{i}}\\
 & =\sum_{i=0}^{t+1}\left(k_{i}-2i\right)2^{1-k_{i}}.\qedhere
\end{align*}
\end{proof}
\begin{proof}[Proof of Claim~\ref{clm:influence-max}]
(A full computation of $\I\left[\ell\left\langle \mu\right\rangle \right]$
for the case $\mu<\frac{1}{4}$) Recall that $k_{i}=i+3$ for $i<j$
and $k_{i}\ge i+4$ for $i\ge j$. Thus,
\begin{align*}
\I\left[\ell\left\langle \mu\right\rangle \right] & =\sum_{i=0}^{\infty}\left(k_{i}-2i\right)2^{1-k_{i}}\\
 & =\sum_{i=0}^{j-1}\left(i+3-2i\right)2^{1-i-3}+\sum_{i=j}^{\infty}\left(k_{i}-2i\right)2^{-1-k_{i}}\\
 & \le\frac{1}{4}\sum_{i=0}^{j-1}\left(3-i\right)2^{-i}+2\sum_{i=j}^{\infty}\left(k_{i}-2j\right)2^{-k_{i}}\\
 & \le\frac{3}{4}\sum_{i=0}^{j-1}2^{-i}-\frac{1}{4}\sum_{i=0}^{j-1}i2^{-i}+2\sum_{i=j+4}^{\infty}\left(i-2j\right)2^{-i}\\
 & =\frac{3}{2}\left(1-2^{-j}\right)-\frac{1}{4}\left(2-\sum_{i=j}^{\infty}i2^{-i}\right)+2\sum_{i=j+4}^{\infty}i2^{-i}-4j\sum_{i=j+4}^{\infty}2^{-i}\\
 & =1-3\cdot2^{-j-1}+\frac{1}{4}\left(j+1\right)2^{1-j}+2\left(j+4+1\right)2^{1-j-4}-4j\cdot2^{1-j-4}\\
 & =1+2^{-j-2}\left(j+1\right)\le\frac{5}{4}<\frac{4}{3},
\end{align*}
where in the second to last equality we used the identity $\sum_{i=j}^{\infty}i2^{-i}=\left(j+1\right)2^{1-j}$.
\end{proof}
\begin{proof}[Proof of Claim~\ref{clm:psi-balanced}]
We have 
\begin{align*}
2-2h\left(p/2\right) & =2+2\frac{p}{2}\log_{2}\frac{p}{2}+2\left(1-\frac{p}{2}\right)\log_{2}\left(1-\frac{p}{2}\right)\\
 & =p\log_{2}p+\left(2-p\right)\log_{2}\left(2-p\right)
\end{align*}
so
\begin{align*}
\psi\left(p,1/2\right)= & \tilde{h}\left(p/2\right)+2p\left(h\left(p\right)+1-h\left(p/2\right)\right)\\
= & -\left(1-p\right)^{2}\log_{2}\left(\left(1-p\right)^{2}\right)-\left(2p-p^{2}\right)\log_{2}\left(2p-p^{2}\right)\\
 & +2ph\left(p\right)+p^{2}\log_{2}p+p\left(2-p\right)\log_{2}\left(2-p\right)\\
= & 2ph\left(p\right)-2\left(1-p\right)^{2}\log_{2}\left(1-p\right)-p\left(2-p\right)\log_{2}p+p^{2}\log_{2}p\\
= & 2ph\left(p\right)-2\left(1-p\right)\left(\left(1-p\right)\log_{2}\left(1-p\right)+p\log_{2}p\right)=2h\left(p\right).\qedhere
\end{align*}
\end{proof}
\medskip{}

A full computation of $H_{*}$:
\begin{align*}
H_{*} & =\H\left[g\right]=\H\left[{\rm Or}_{2}\sqcap G\right]\\
 & =\frac{2}{3}\left(\H\left[{\rm Or}_{2}\right]-\tilde{h}\left(3/4\right)\right)+\frac{3}{4}\left(\H\left[G\right]-\tilde{h}\left(2/3\right)\right)+\psi\left(3/4,2/3\right)\\
 & =\frac{2}{3}\left(2-h\left(1/4\right)\right)+\frac{3}{4}\left(2\log_{2}3-h\left(1/9\right)\right)+\tilde{h}\left(1/2\right)+2\left(h\left(3/4\right)+h\left(2/3\right)-h\left(1/2\right)\right)\\
 & =\frac{4}{3}+\left(2-\frac{2}{3}\right)h\left(1/4\right)+\frac{3}{4}\cdot\frac{8}{3}+0+2h\left(1/3\right)-2\cdot1\\
 & =\frac{4}{3}+\frac{4}{3}\left(2-\frac{3}{4}\log_{2}3\right)+2\left(\log_{2}3-\frac{2}{3}\right)=\frac{8}{3}+\log_{2}3.\qedhere
\end{align*}

\subsection{From Section~\ref{sec:LB-2}}

A full computation of $\tilde{\H}\left[\tau\right]$ for the bias
$\eta=1-2\Phi$:
\begin{align*}
\tilde{\H}\left[\tau\right] & =-\Phi^{6}\log_{2}\Phi^{6}-2\cdot4\Phi^{5}\log_{2}(4\Phi^{5})-4\Phi^{6}\log_{2}(4\Phi^{6})\\
 & =-\Phi^{5}\left(6\Phi\log_{2}\Phi+40\log_{2}\Phi+16+8\Phi+24\Phi\log_{2}\Phi\right)\\
 & =-8\Phi^{5}\left(2+\Phi\right)-10\left(3\Phi+4\right)\Phi^{5}\log_{2}\Phi\\
 & =8\left(1-2\Phi\right)+10\left(4\Phi-3\right)\log_{2}\Phi.\qedhere
\end{align*}

\begin{proof}[Proof of Claim~\ref{clm:Fm-ratio}]
Recall that earlier we computed $\tilde{\I}\left[\tau\right]=8\Phi^{4}$.
By Lemma~\ref{lem:composition}, 
\begin{align*}
\I\left[F_{m+1}\right] & =\tilde{\I}\left[\tau\right]\I^{+}\left[F_{m}\right]=\tilde{\I}\left[\tau\right]\frac{\I\left[F_{m}\right]}{\V\left[F_{m}\right]}\\
 & =\frac{8\Phi^{4}}{4\Phi^{3}}\cdot\I\left[F_{m}\right]=2\Phi\cdot\I\left[F_{m}\right].
\end{align*}
Furthermore,
\begin{align*}
\frac{\H^{+}\left[F_{m+1}\right]}{\I^{+}\left[F_{m+1}\right]} & =\frac{\H\left[F_{m+1}\right]-\tilde{h}\left(\Phi\right)}{\I\left[F_{m+1}\right]}\\
 & =\frac{\tilde{\I}\left[\tau\right]\cdot\H^{+}\left[F_{m}\right]}{\tilde{\I}\left[\tau\right]\cdot\I^{+}\left[F_{m}\right]}+\frac{\tilde{\H}\left[\tau\right]-\tilde{h}\left(\Phi\right)}{\I\left[F_{m+1}\right]},
\end{align*}
hence
\begin{align*}
\frac{\H^{+}\left[F_{m}\right]}{\I^{+}\left[F_{m}\right]} & =\frac{\H^{+}\left[F_{0}\right]}{\I^{+}\left[F_{0}\right]}+\left(\tilde{\H}\left[\tau\right]-\tilde{h}\left(\Phi\right)\right)\sum_{k=1}^{m}\frac{1}{\I\left[F_{k}\right]}\\
 & =\frac{\H^{+}\left[F_{0}\right]}{\I^{+}\left[F_{0}\right]}+\frac{\tilde{\H}\left[\tau\right]-\tilde{h}\left(\Phi\right)}{\I\left[F_{0}\right]}\sum_{k=1}^{m}\left(2\Phi\right)^{-k}\\
 & =\frac{\H^{+}\left[\ell\left\langle \Phi\right\rangle \right]}{\I^{+}\left[\ell\left\langle \Phi\right\rangle \right]}+\frac{\tilde{\H}\left[\tau\right]-\tilde{h}\left(\Phi\right)}{\I\left[\ell\left\langle \Phi\right\rangle \right]}\cdot\frac{1-\left(2\Phi\right)^{-m}}{2\Phi-1}.
\end{align*}
Asymptotically the $\left(2\Phi\right)^{-m}$ term disappears, and
we have 
\begin{align*}
\lim_{m\to\infty}\frac{\H\left[F_{m}\right]}{\I\left[F_{m}\right]} & =\lim_{m\to\infty}\frac{\H^{+}\left[F_{m}\right]-\tilde{h}\left(\Phi\right)}{\I^{+}\left[F_{m}\right]}=\lim_{m\to\infty}\frac{\H^{+}\left[F_{m}\right]}{\I^{+}\left[F_{m}\right]}-0\\
 & =\frac{\H^{+}\left[\ell\left\langle \Phi\right\rangle \right]}{\I^{+}\left[\ell\left\langle \Phi\right\rangle \right]}+\frac{\tilde{\H}\left[\tau\right]-\tilde{h}\left(\Phi\right)}{\left(2\Phi-1\right)\I\left[\ell\left\langle \Phi\right\rangle \right]}\\
 & =\frac{\H\left[\ell\left\langle \Phi\right\rangle \right]-\tilde{h}\left(\Phi\right)}{\I\left[\ell\left\langle \Phi\right\rangle \right]}+\frac{\tilde{\H}\left[\tau\right]-\tilde{h}\left(\Phi\right)}{\left(2\Phi-1\right)\I\left[\ell\left\langle \Phi\right\rangle \right]}\\
 & =\frac{\H\left[\ell\left\langle \Phi\right\rangle \right]+\left(3+2\Phi\right)\tilde{\H}\left[\tau\right]-\left(4+2\Phi\right)\tilde{h}\left(\Phi\right)}{\I\left[\ell\left\langle \Phi\right\rangle \right]}.\qedhere
\end{align*}
\end{proof}

\subsection{From Section~\ref{sec:LB-3}}
\begin{proof}[Proof of Proposition~\ref{prop:step}]
We have 
\[
\Pr\left[\kappa\right]=1-\Pr\left[\lambda\sqcap\kappa\right]=1-p\Pr\left[\kappa\right],
\]
so we can solve for 
\[
q=\Pr\left[\kappa\right]=\frac{1}{1+p}.
\]
Next, by Proposition~\ref{prop:sqcap} and duality we have
\begin{align*}
\I\left[\kappa\right] & =\I\left[\lambda\sqcap\kappa\right]=q\I\left[\lambda\right]+p\I\left[\kappa\right]
\end{align*}
and we can solve for $\I\left[\kappa\right]=q\I\left[\lambda\right]/\left(1-p\right),$yielding
\[
\I^{+}\left[\kappa\right]=\frac{\I\left[\kappa\right]}{4q\left(1-q\right)}=\frac{q\I\left[\lambda\right]}{4q\cdot pq\cdot\left(1-p\right)}=\frac{1}{q}\I^{+}\left[\lambda\right].
\]
Last, we compute $\H\left[\kappa\right]$. Note that since $q=1-pq$
we have $h\left(pq\right)=h\left(q\right)$ and $\tilde{h}\left(pq\right)=\tilde{h}\left(q\right)$,
so
\begin{align*}
\psi\left(p,q\right) & =\tilde{h}\left(pq\right)+4pq\left(h\left(p\right)+h\left(q\right)-h\left(pq\right)\right)\\
 & =\tilde{h}\left(q\right)+4pqh\left(p\right).
\end{align*}
Now, by Proposition~\ref{prop:sqcap},
\begin{align*}
\H\left[\kappa\right] & =\H\left[\lambda\sqcap\kappa\right]=q\left(\H\left[\lambda\right]-\tilde{h}\left(p\right)\right)+p\left(\H\left[\kappa\right]-\tilde{h}\left(q\right)\right)+\psi\left(p,q\right)\\
 & =p\H\left[\kappa\right]+q\left(\H\left[\lambda\right]-\tilde{h}\left(p\right)\right)-p\tilde{h}\left(q\right)+\tilde{h}\left(q\right)+4pqh\left(p\right)
\end{align*}
so we can solve for 
\begin{align*}
\H\left[\kappa\right] & =\frac{q\left(\H\left[\lambda\right]-\tilde{h}\left(p\right)\right)+\left(1-p\right)\tilde{h}\left(q\right)+4pqh\left(p\right)}{1-p}\\
 & =\frac{q}{1-p}\left(\H\left[\lambda\right]-\tilde{h}\left(p\right)+4ph\left(p\right)\right)+\tilde{h}\left(q\right),
\end{align*}
yielding
\begin{align*}
\H^{+}\left[\kappa\right] & =\frac{\H\left[\kappa\right]-\tilde{h}\left(q\right)}{4q\left(1-q\right)}=\frac{\H\left[\lambda\right]-\tilde{h}\left(p\right)+4ph\left(p\right)}{4\left(1-p\right)\cdot pq}\\
 & =\frac{\H\left[\lambda\right]-\tilde{h}\left(p\right)}{4p\left(1-p\right)q}+\frac{h\left(p\right)}{\left(1-p\right)q}\\
 & =\left(\H^{+}\left[\lambda\right]+\frac{h\left(p\right)}{1-p}\right)/q.\qedhere
\end{align*}
\end{proof}
\begin{proof}[Proof of Claim~\ref{clm:qm-linear-rational}]
For $m=0$ indeed $q_{m}=\left(1\cdot z+0\right)/\left(0\cdot z+1\right)=z$.
Now, assuming the claim holds for $q_{m}$, 
\begin{align*}
q_{m+1} & =\frac{1}{1+q_{m}}=\frac{b_{m}z+b_{m+1}}{\left(b_{m}z+b_{m+1}\right)\left(1+q_{m}\right)}\\
 & =\frac{b_{m}z+b_{m+1}}{b_{m}z+b_{m+1}+b_{m-1}z+b_{m}}\\
 & =\frac{b_{m}z+b_{m+1}}{\left(b_{m}+b_{m-1}\right)z+\left(b_{m+1}+b_{m}\right)}\\
 & =\frac{b_{m}z+b_{m+1}}{b_{m+1}z+b_{m+2}}.\qedhere
\end{align*}
\end{proof}
\begin{proof}[Proof of Claim~\ref{clm:Fm-ratio-redux}]
By~(\ref{eq:I-sequence}) and~(\ref{eq:H-sequence}), 
\begin{align*}
\frac{\H^{+}\left[F_{k+1}\right]}{\I^{+}\left[F_{k+1}\right]} & =\frac{\H^{+}\left[F_{k}\right]}{\I^{+}\left[F_{k}\right]}+\frac{h\left(q_{k}\right)}{\left(1-q_{k}\right)\I^{+}\left[F_{k}\right]}\\
 & =\frac{\H^{+}\left[F_{k}\right]}{\I^{+}\left[F_{k}\right]}+\frac{h\left(q_{k}\right)\pi_{k}}{\left(1-q_{k}\right)\I^{+}\left[F_{0}\right]}\\
 & =\frac{\H^{+}\left[F_{k}\right]}{\I^{+}\left[F_{k}\right]}+\frac{h\left(q_{k}\right)\pi_{k-2}}{\I^{+}\left[F_{0}\right]},
\end{align*}
thus 
\begin{align*}
\frac{\H^{+}\left[F_{m}\right]}{\I^{+}\left[F_{m}\right]} & =\frac{\H^{+}\left[F_{0}\right]}{\I^{+}\left[F_{0}\right]}+\sum_{k=0}^{m-1}\frac{\H^{+}\left[F_{k+1}\right]}{\I^{+}\left[F_{k+1}\right]}-\frac{\H^{+}\left[F_{k}\right]}{\I^{+}\left[F_{k}\right]}\\
 & =\frac{\H^{+}\left[F_{0}\right]}{\I^{+}\left[F_{0}\right]}+\sum_{k=-2}^{m-3}\frac{h\left(q_{k+2}\right)\pi_{k}}{\I^{+}\left[F_{0}\right]}
\end{align*}
and
\begin{align*}
\frac{\H\left[F_{m}\right]}{\I\left[F_{m}\right]} & =\frac{\H^{+}\left[F_{m}\right]}{\I^{+}\left[F_{m}\right]}+\frac{\tilde{h}\left(q_{m}\right)}{\I\left[F_{m}\right]}\\
 & =\frac{\H^{+}\left[F_{m}\right]}{\I^{+}\left[F_{m}\right]}+\frac{\tilde{h}\left(q_{m}\right)}{4q_{m}\left(1-q_{m}\right)\I^{+}\left[F_{m}\right]}\\
 & =\frac{\H^{+}\left[F_{0}\right]}{\I^{+}\left[F_{0}\right]}+\frac{\tilde{h}\left(q_{m}\right)\pi_{m-2}}{4q_{m}\I^{+}\left[F_{0}\right]}+\sum_{k=-2}^{m-3}\frac{h\left(q_{k+2}\right)\pi_{k}}{\I^{+}\left[F_{0}\right]}\\
 & =\frac{H-\tilde{h}\left(z\right)}{I}+\frac{4z\left(1-z\right)}{I}\left(\frac{1+q_{m-1}}{4}\tilde{h}\left(q_{m}\right)\pi_{m-2}+\sum_{k=-2}^{m-3}h\left(q_{k+2}\right)\pi_{k}\right)\cdot\\
 & =\frac{1}{I}\left(H-\tilde{h}\left(z\right)+z\left(1-z\right)\left(\pi_{m-1}+\pi_{m-2}\right)\tilde{h}\left(q_{m}\right)+\beta_{m}\left(z\right)\right).\qedhere
\end{align*}
\end{proof}
\pagebreak{}

\subsubsection{Lower bound obtained from $T=\left(T\sqcap T\right)^{\dagger}$}

First we prove an analogue of Proposition~\ref{prop:step}: 

Given a Boolean function $\lambda$, define $\kappa=\left(\lambda\sqcap\lambda\right)^{\dagger}$.
Writing $p=\Pr\left[\lambda\right]$ and $q=\Pr\left[\kappa\right]$
we have $q=1-p^{2}.$ By Proposition~\ref{prop:sqcap} we have 
\begin{align*}
\I\left[\kappa\right] & =\I\left[\lambda\sqcap\lambda\right]=2p\I\left[\lambda\right],\\
\H\left[\kappa\right] & =\H\left[\lambda\sqcap\lambda\right]=2p\left(\H\left[\lambda\right]-\tilde{h}\left(p\right)\right)+\psi\left(p,p\right)\\
 & =2p\left(\H\left[\lambda\right]-\tilde{h}\left(p\right)\right)+\tilde{h}\left(p^{2}\right)+4p^{2}\left(2h\left(p\right)-h\left(p^{2}\right)\right)\\
 & =2p\left(\H\left[\lambda\right]-\tilde{h}\left(p\right)\right)+\tilde{h}\left(q\right)+8p^{2}h\left(p\right)-4p^{2}h\left(q\right).
\end{align*}
Let $\tilde{C}_{m}=\left(\H\left[T_{m}\right]-\tilde{h}\left(t_{m}\right)\right)/\I\left[T_{m}\right]$,
where $t_{m}=\Pr\left[T_{m}\right]$ . Now
\begin{align*}
\tilde{C}_{m+1}-\tilde{C}_{m} & =\frac{8p^{2}h\left(p\right)-4p^{2}h\left(q\right)}{\I\left[T_{m+1}\right]}=\frac{4ph\left(p\right)}{\I\left[T_{m}\right]}-\frac{4\left(1-q\right)h\left(q\right)}{\I\left[T_{m+1}\right]}
\end{align*}
so
\begin{align*}
\lim_{m\to\infty}\frac{\H\left[T_{m}\right]}{\I\left[T_{m}\right]} & =\tilde{C}_{0}+\sum_{k=0}^{\infty}\left(\tilde{C}_{k+1}-\tilde{C}_{k}\right)=\frac{\H\left[T_{0}\right]-\tilde{h}\left(t_{0}\right)+\gamma\left(t_{0}\right)}{\I\left[T_{0}\right]},\\
\end{align*}
where 
\begin{align*}
\gamma\left(z\right) & =4\sum_{k=0}^{\infty}\frac{t_{k}h\left(t_{k}\right)}{2^{k}\prod_{i=0}^{k-1}t_{i}}-4\sum_{k=1}^{\infty}\frac{\left(1-t_{k}\right)h\left(t_{k}\right)}{2^{k}\prod_{i=0}^{k-1}t_{i}}\\
 & =4zh\left(z\right)+4\sum_{k=1}^{\infty}\frac{\left(2t_{k}-1\right)h\left(t_{k}\right)}{2^{k}\prod_{i=0}^{k-1}t_{i}}.
\end{align*}

\subsection{From Section~\ref{sec:Lipschitz}}

The difference in entropies is:
\begin{align*}
\H\left[f\right]-\H\left[g\right] & =\sum_{S}\hat{g}\left(S\right)^{2}\log_{2}\left(\hat{g}\left(S\right)^{2}\right)-\hat{f}\left(S\right)^{2}\log_{2}\left(\hat{f}\left(S\right)^{2}\right)\\
 & =\frac{2}{N^{2}}\sum_{S}\left(N\hat{g}\left(S\right)\right)^{2}\log_{2}\left|\hat{g}\left(S\right)\right|-(N\hat{f}\left(S\right))^{2}\log_{2}|\hat{f}\left(S\right)|\\
 & =\frac{2}{N^{2}}\sum_{S}\Big[\left(\hat{a}\left(S\right)+1\right)^{2}\log_{2}\left(\left|\hat{a}\left(S\right)+1\right|/N\right)\\
 & \qquad\quad-\left(\hat{a}\left(S\right)-1\right)^{2}\log_{2}\left(\left|\hat{a}\left(S\right)-1\right|/N\right)\Big]\\
 & =\frac{2}{N^{2}}\sum_{S}\Big[\left(1+\hat{a}\left(S\right)\right)^{2}\log_{2}\left(\left|1+\hat{a}\left(S\right)\right|/N\right)\\
 & \qquad\quad-\left(1-\hat{a}\left(S\right)\right)^{2}\log_{2}\left(\left|1-\hat{a}\left(S\right)\right|/N\right)\Big]\\
 & =\frac{2}{N^{2}}\sum_{k=1}^{N/2}\Delta_{k}\left[\left(2k\right)^{2}\log_{2}\left(2k/N\right)-\left(2k-2\right)^{2}\log_{2}\left(2\left(k-1\right)/N\right)\right]\\
 & =\frac{8}{N^{2}}\sum_{k=1}^{N/2}\Delta_{k}\left[k^{2}\log_{2}k-\left(k-1\right)^{2}\log_{2}\left(k-1\right)\right]\\
 & \qquad-\frac{8\left(n-1\right)}{N^{2}}\sum_{k=1}^{N/2}\Delta_{k}\left[k^{2}-\left(k-1\right)^{2}\right]\\
\text{[by\,\eqref{eq:delta-sum}]} & =\frac{8}{N^{2}}\sum_{k=1}^{N/2}\Delta_{k}\left[k^{2}\log_{2}k-\left(k-1\right)^{2}\log_{2}\left(k-1\right)\right]\\
\text{[note the index \ensuremath{k}]} & =\frac{8}{N^{2}}\sum_{k=2}^{N/2}\Delta_{k}\left[k^{2}\log_{2}k-\left(k-1\right)^{2}\log_{2}\left(k-1\right)\right]\\
 & =\frac{8}{N^{2}}\sum_{k=2}^{N/2}\Delta_{k}k^{2}\log_{2}\frac{k}{k-1}+\frac{8}{N^{2}}\sum_{k=2}^{N/2}\Delta_{k}\left(2k-1\right)\log_{2}\left(k-1\right).
\end{align*}

\end{document}